\documentclass{lmcs} %%% last changed 2014-08-20


\usepackage{tikz-cd}
\usepackage{amssymb}
\usepackage{mathpartir}
\usepackage{etoolbox}
\usepackage[hidelinks]{hyperref}
\usepackage{zref-clever}

\zcsetup{
  cap,              % Capitalize all labels (Section, Figure, Table)
  abbrev = false,   % Avoid abbreviations (e.g., Figure instead of Fig.)
  nameinlink = false,  % Makes ONLY the reference number clickable
  countertype = {
    thm = thm,
    cor = cor,
    lem = lem,
    prop = prop,
    defi = defi,
    rem  = rem,
    exa  = exa,
  }
}

\zcRefTypeSetup{thm}{ Name-sg = Theorem, Name-pl = Theorems }
\zcRefTypeSetup{cor}{ Name-sg = Corollary, Name-pl = Corollaries }
\zcRefTypeSetup{lem}{ Name-sg = Lemma, Name-pl = Lemmas}
\zcRefTypeSetup{prop}{ Name-sg = Proposition, Name-pl = Propositions}
\zcRefTypeSetup{defi}{ Name-sg = Definition, Name-pl = Definitions}
\zcRefTypeSetup{rem}{ Name-sg = Remark, Name-pl = Remarks}
\zcRefTypeSetup{exa}{ Name-sg = Example, Name-pl = Examples}

\DeclareMathAlphabet{\mathpzc}{OT1}{pzc}{m}{it}

\mathcode`\:="603A     % : = \colon
\mathchardef\colon="303A  % :=
\newcommand{\typing}{\colon}

\newcommand{\define}[1]{\emph{#1}}
\newcommand{\arcmp}{\circ}

\newcommand{\rwhisk}[1]{#1} %writes the 1-cell to the right
\newcommand{\lwhisk}[1]{#1} %writes the 2-cell to the right
\newcommand{\mate}[1]{#1^{\#}} 

\def\ct#1{\ensuremath{\mathpzc{#1}}}

\def\cttwo#1{\ensuremath{\mathbf{#1}}}
\def\op{\ensuremath{^{\mathrm{op}}}}

\def\rstn{\ensuremath{\mathord{\restriction}}}
\def\gint{\ensuremath{\mathord{\int}\kern-.2ex}}
\def\Hom#1{\ensuremath{#1}}
\def\OPR#1{\ensuremath{\mathop{\mathrm{#1}}\nolimits}} % perché \mathop?
\def\Obj{\OPR{Obj}}

\def\Arr#1{\ensuremath{#1^{2}}}
\def\cod{\OPR{cod}}
\def\dom{\OPR{dom}}

\def\id#1{\ensuremath{\mathrm{id}_{#1}}} %identity 1-cells
\def\idtwo#1{\ensuremath{\mathrm{i}_{#1}}} %identity 2-cells
\def\Id#1{\ensuremath{\mathrm{Id}_{#1}}} %identity functors between categories
\def\twoId#1{\ensuremath{\boldsymbol{\mathrm{Id}_{#1}}}} %identity functors between 2-categories

\newcommand{\obzero}[1]{\ensuremath{0}^#1}

\def\fp#1{\ensuremath{|#1|}}  %components of arrows in jacobs completion
\def\sp#1{\ensuremath{\num{#1}}}
\def\tp#1{\ensuremath{\widetilde{#1}}}

\def\ini#1#2{\ensuremath{#1\hookrightarrow #2}}  %initial segment inclusion
\def\carr#1{\ensuremath{#1^+}}  %comprehension arrow
\def\Cat{\ensuremath{\ct{Cat}}}  %categoria cat
\def\Cattwo{\ensuremath{\cttwo{Cat}}}  %2-categoria cat
\newcommand{\ctset}{\ct{Set}}  %categoria set
\newcommand{\ctsetp}{\ct{Set}_*}  %categoria set puntata
\newcommand{\Fam}[1]{\mathrm{Fam}(#1)}
\newcommand{\Famfib}[1]{\mathrm{Fam}_{#1}}
\newcommand{\tfib}[1]{\dot{#1}}  %term fibration
\newcommand{\tcat}[1]{\dot{#1}}  %term category
\newcommand{\tsigma}{\dot{\Sigma}}  %term sigma
\newcommand{\gtfib}[1]{\tilde{#1}}  %global type morphism fibration
\newcommand{\gtcat}[1]{\tilde{#1}}  %global type morphism category
\newcommand{\gtsigma}{\tilde{\Sigma}}  %global type morphism sigma

\newcommand{\ftransp}{\#}  %transposition functor on fibrations
\newcommand{\ttransp}{\#}  %transposition functor on total categories

\newcommand{\nt}{\Rightarrow} %natural transformations
\newcommand{\twofun}[1]{\mathbf{#1}}
\newcommand{\twont}[1]{\boldsymbol{#1}}
\newcommand{\forgT}{\twofun{U_T}} %forgets terminal JTComp -> JComp
\newcommand{\forgLE}{\twofun{U_{LE}}} %forgets LE LECC -> JTComp
\newcommand{\forgJ}{\twofun{U_J}} %forgets J JComp -> Fib
\newcommand{\reind}[1]{#1^*} %reindex functor along argument
\newcommand{\carar}[2]{#1^{#2}} %cartesian lifting of the argument 1 at the argument 2

\newcommand{\num}[1]{\underline{#1}}

\def\Fib{\ensuremath{\cttwo{Fib}}}
\newcommand{\Fibone}{\ct{Fib}}
\newcommand{\Fibdisc}{\ct{Fib}_d}

\newcommand{\fb}[1]{{#1}^{\mathrm{b}}} %freccia sotto in quadrati
\newcommand{\ft}[1]{{#1}^{\mathrm{t}}} %freccia sopra in quadrati
\def\FFFP#1{\ensuremath{\ct{FFP}(#1)}}  %category

\def\fffp#1{\ensuremath{\mathbf{ffp}(#1)}}  %fibration
\def\Fffp{\ensuremath{\mathbf{ffp}}}  %functor, but actually used for the fibration
\newcommand{\Fp}{\ensuremath{\mathbf{ffp}}}  %2-functor

\def\LComp{\ensuremath{\cttwo{LECC}}} %Lawvere comprehension category
\def\JComp{\ensuremath{\cttwo{JCC}}} %Jacobs comprehension category
\def\JTComp{\ensuremath{\cttwo{JCC_T}}} %Jacobs comp. cat. w. fib. term.

\newcommand{\unitty}{\ensuremath{\top}}
\newcommand{\unittr}{\ensuremath{*}}
\newcommand{\type}{\ensuremath{\mathrm{Type}}}
\newcommand{\ctx}{\ensuremath{\mathrm{Ctx}}}

\def\LFib#1{\ensuremath{\hat{#1}}}  %fibrazione Lawvere libera
\def\LFun#1{\ensuremath{\hat{#1}}}  %azione del funtore libero su frecce
\newcommand{\term}[1]{\ensuremath{\mathrm{T}^{#1}}} %funtore terminale fibrazione
\newcommand{\comp}[1]{\ensuremath{\mathrm{C}^{#1}}} %funtore comprensione fibrazione
\newcommand{\coun}[1]{\ensuremath{\mathrm{\epsilon}^{#1}}} %counità in aggiunzione di Lawvere
\newcommand{\un}[1]{\ensuremath{\mathrm{\eta}^{#1}}} %unità in aggiunzione di Lawvere
\newcommand{\comon}[1]{\OPR{T}}  %fibred comonad of products
\newcommand{\LtoJ}{\ensuremath{\twofun{LEJ}}}

\def\JDom#1#2{\ensuremath{{#1}^{\twofun{J}}_{#2}}}  %dominio fibrazione Jacobs libera
\def\JTDom#1#2{\ensuremath{{#1}^{\twofun{T}}_{#2}}}  %domain free comp. cat. w. fib. prod.
\def\JDomE#1{\JDom{\ct{E}}{#1}}  %dominio fibrazione Jacobs libera
\def\JTDomE#1{\JTDom{\ct{E}}{#1}}  %domain free comp. cat. w. fib. prod.
\def\JFib#1{\ensuremath{{#1}^{\twofun{J}}}}  %fibrazione Jacobs libera
\def\JTFib#1{\ensuremath{{#1}^{\twofun{T}}}}  %free comprehension preserves fib. term.
\def\JFun#1{\ensuremath{{#1}^{\twofun{J}}}}  %azione del funtore libero su frecce

\def\JTFun#1{\ensuremath{{#1}^{\twofun{T}}}}
\def\Jun#1{\ensuremath{\twont{\eta}^{\twofun{J}}_{#1}}}  
\def\JTun#1{\ensuremath{\twont{\eta}^{\twofun{T}}_{#1}}}  
\def\Jcoun#1{\ensuremath{\twont{\epsilon}^{\twofun{J}}_{#1}}}
\def\JTcoun#1{\ensuremath{\twont{\epsilon}^{\twofun{T}}_{#1}}}
\def\JFree{\ensuremath{\twofun{F^J}}}  %funtore libero da Fib a JComp
\def\JTFree{\ensuremath{\twofun{F^T}}}  %free functor from JComp to JTComp

\def\C#1#2{\ensuremath{c^{#1}_{#2}}}  %frecce c_i^k
\newcommand{\jcomp}[1]{\chi^{#1}} %comprensione di Lawvere
\def\JFC#1{\ensuremath{\jcomp{\JFib{#1}}}}  %funtore comprensione libero su fibrazione
\newcommand{\weak}[2]{\mathrm{w}_{#1}^*#2}  %weakening di 2 lungo comprensione di 1
\newcommand{\weakar}[2]{\mathrm{w}_{#1}#2}  %weakening arrow di 2 lungo comprensione di 1
\newcommand{\gel}[1]{\mathrm{g}_{#1}}  %elemento generico su 1
\newcommand{\var}[1]{x}  %variabile generica

\newcommand{\Tr}{\OPR{Tr}}  %category pre-quotient
\newcommand{\Trfib}[1]{\OPR{Trfib} #1}  %category pre-quotient
\newcommand{\Const}{\OPR{Const}}  %constants functor
\newcommand{\ConstCat}{\OPR{ConstC}}  %constants functor
\newcommand{\ConstFib}[1]{\overline{#1}}  %functor from constcat to base
\newcommand{\triang}{tr}  %constant to triangles functor
\newcommand{\incl}{incl}  %constant to old morphisms
\newcommand{\LEFree}{\twofun{F^{LE}}}  
\newcommand{\LEFib}[1]{{#1}^{\twofun{LE}}}
\newcommand{\LEFun}[1]{{#1}^{\twofun{LE}}}
\newcommand{\LEDom}[2]{{#1}^{\twofun{LE}}_{#2}}
\newcommand{\LEDomE}[1]{\LEDom{\ct{E}}{#1}}
\newcommand{\LEun}[1]{\twont{\eta}^{\twofun{LE}}_{#1}}
\newcommand{\LEcoun}[1]{\twont{\epsilon}^{\twofun{LE}}_{#1}}
\newcommand{\quot}[1]{\OPR{Q}^{#1}}  %quotient towards LEFib
\newtheorem*{notation}{Notation}

\newcommand{\fibre}[3][]{#2\ifblank{#1}{}{^{#1}}_{#3}} 

\theoremstyle{plain} %\zcrefname{satz}{Satz}{S\"atze}
\def\eg{{\em e.g.}}

\tikzset{Rightarrow/.style={double equal sign distance,>={Implies},->},
	triple/.style={-,preaction={draw,Rightarrow}}}

\begin{document}

% If the title is longer than 55 characters, then specify a shorter running title as the optional argument to \title. The running title should be roughyl at most 55 characters:
\title[Free constructions for comprehension categories]{Free constructions for comprehension categories}
% \titlecomment{{\lsuper*}OPTIONAL comment concerning the title, \eg,
% 	if a variant or an extended abstract of the paper has appeared elsewhere.}
% \thanks{thanks, optional.}	%optional

% affiliations are numbered automatically with a, b, c (see below)
% use the optional argument to indicate the affiliation(s) of each author
% omit the argument if there is only one author, or only one affiliation
\author[F.~Dagnino]{Francesco Dagnino}[a]
\author[J.~Emmenegger]{Jacopo Emmenegger}[b]
\author[A.~Giusto]{Andrea Giusto}[a]

% affiliation 1 (automatically numbered a)
\address{University of Genoa, DIBRIS, Italy}	%optional
% write emails for all authors having that affiliation
\email{francesco.dagnino@unige.it, andrea.giusto@edu.unige.it}  %optional

% affiliation 2 (automatically numbered b)
\address{University of Genoa, DIMA, Italy}	%optional
\email{jacopo.emmenegger@gmail.com}  %optional

%% etc.

%% required for running head on odd and even pages, use suitable
%% abbreviations in case of long titles and many authors:

%%%%%%%%%%%%%%%%%%%%%%%%%%%%%%%%%%%%%%%%%%%%%%%%%%%%%%%%%%%%%%%%%%%%%%%%%%%

%% the abstract has to PRECEDE the command \maketitle:
%% be sure not to issue the \maketitle command twice!

\begin{abstract}
	\noindent 
Jacobs comprehension categories subsume a large class of categorical models of type dependency, 
supporting also the description of morphisms between types.
We study the relationship between comprehension categories and a particular subclass, which we call Lawvere-Ehrhard comprehension categories. 
First, we characterize this subclass by comparing a fibration of terms and a fibration of type morphisms associated to a given comprehension category. 
Next, we provide the construction of the free comprehension category over a fibration. 
Finally, we construct the free Lawvere-Ehrhard comprehension category over a Jacobs comprehension category.
\end{abstract}

\maketitle

%% start the paper here:
\section{Introduction}\label{S:one}

Type dependency has been fruitfully studied using category theory, leading to the introduction of a wide range of models such as
contextual categories \cite{cartmell1986generalised}, categories with families \cite{dybjer1995internal}, natural models \cite{awodey2018natural}, and many others. 
All these structures are essentially built around three  cornerstones: 
a category of contexts and substitutions, 
for every context $\Gamma$ a collection of types depending on $\Gamma$, and 
operations capturing the substitution and context extension rules, i.e. 
\begin{mathpar}
\inferrule*{
  \sigma : \Delta\to\Gamma \\
  \Gamma \vdash A\ \type 
}{ \Delta \vdash A[\sigma]\ \type }
\and 
\inferrule*{
  \Gamma\vdash A \ \type
}{ \vdash \Gamma,\, x\typing A \ \ctx}
\end{mathpar}

In \cite{ahrens2024comparing} the authors give a summary of the relationships between these and more structures,
recognizing that all the models can be described as comprehension categories \cite{JACOBS1993169}, 
which thus provide a unifying framework for studying type dependency. 
Essentially, these are (Grothendieck) fibrations with additional structure specifically designed to capture the context extension rule. 

Fibrations provide a compact and manageable way of describing families of categories indexed by a base category: 
the categories in the family are packed into a single category, the total category, together with a functor into the base category whose fibres are the original categories,
and the action of the arrows in the base is encoded by making sure that the total category contains enough arrows, called cartesian, that satisfy a certain universal property with respect to the functor into the base.
A \emph{comprehension category} consists of a fibration $p: \ct{E}\to \ct{B}$ together with a functor over $\ct{B}$, often written $\chi$, that assigns an arrow in $\ct{B}$ to every object in $\ct{E}$, and commutative squares in $\ct{B}$ to arrows in $\ct{E}$, in such a way that cartesian arrows are mapped to pullback squares.

There is a class of examples arising from dependent type theories:
objects and arrows in the base category represent (telescopic) contexts and substitutions, respectively;
the fibre over each context $\Gamma$ is the set of types depending on $\Gamma$;
and cartesian arrows are there to witness the fact that a certain type is obtained by applying an appropriate substitution to another type, exactly one arrow for each such occurrence.
The comprehension structure $\chi$ is obtained from the operation of context extension: the functor assigns to a type $A$ in context $\Gamma$ the substitution $\Gamma, x:A\to \Gamma$ that forgets the variable of type $A$ and fixes all the others, sometimes called display map.
Under this assignment, cartesian arrows are mapped to pullback squares involving two parallel display maps and this is, one could say, the main reason for the definition of comprehension category and of its instances mentioned above.

For usual dependent type theories, the collection of types over a given context is a set rather than a category, as in the example above.
As a consequence, fibrations modelling them have a special property: they are \emph{discrete}, meaning that the fibres are discrete categories, i.e.~sets. 
The reason is that usually there is no primitive syntactic notion of ``type morphism''. 
However, recent works \cite{adjedj2026adaptt, coraglia2023categorical, najmaei2026semantics} highlight the importance of (non-trivial) type morphisms as a way of incorporating forms of subtyping into a dependent type theory, thus requiring arbitrary (not necessarily discrete) comprehension categories for modelling them. 
More precisely, in \cite{coraglia2023categorical} the authors show that faithful comprehension categories, that is, those whose fibres are preorders, 
naturally support a notion of coercive subtyping in the sense of \cite{luo2013coercive}: 
the unique type morphism from a type $A$ to another type $B$ over a context $\Gamma$ is the coercion  realizing  the subtyping relation $A\leq B$. 
Following this perspective, they observe that arbitrary comprehension categories can interpret 
a generalized form of coercive subtyping that is ``proof-relevant'', in the sense that the same subtyping relation can be realized by different coercions. 
In the same spirit, \cite{najmaei2026semantics, adjedj2026adaptt} introduced new dependent type theories where type morphisms are first-class citizens, having a dedicated judgement to construct them, which thus need to be modelled by unrestricted comprehension categories. 

Comprehension categories are designed precisely to mirror in a categorical framework the standard structural rules of a dependent type theory without anything else. 
As a result, they provide a very general framework, which however may be quite wild, as the behaviour of extended contexts having few constraints is weakly characterized. 
Ehrhard~\cite{ehrhard1988} followed a different strategy:
he introduced another class of models of type dependency based on fibrations, called D-categories, aimed at characterizing extended contexts by a universal property. 
In type-theoretic terms, D-categories are fibrations with a unit type $\top$ in every contexts and, for every type $A$ in context $\Gamma$, the extended context $\Gamma, x : A$ is characterized as the one such that substitutions from another context $\Delta$ into $\Gamma,x:A$ are in one to one correspondence with  type morphisms over $\Delta$ from $\top$ into an appropriate substitution of $A$. 
In technical terms, the extended context is a representing object for the above type morphisms.
Note that this definition relies on type morphisms and so, except for degenerate cases, it requires non-discrete fibrations. 
Although apparently different, D-categories and comprehension categories are tightly related. 
Specifically, Jacobs \cite{JACOBS1993169} proved that any D-category gives rise to a comprehension category on the same fibration, showing that Ehrhard's characterization of extended contexts is compatible with the structural rules of dependent type theories.

Another interesting fact about D-categories is that they create a bridge between the context extension rule of dependent type theories and the comprehension schema of logic. 
Indeed, D-categories generalise a categorical description of the set-theoretic comprehension schema given earlier by Lawvere~\cite{MR0257175}, as shown by Jacobs~\cite[Example 4.18]{JACOBS1993169}.
For this reason, from now on we will refer to D-categories as \emph{Lawvere-Ehrhard comprehension categories}, abbreviated LECC. 

The purpose of this paper is to study in detail the relationships between these two categorical approaches to comprehensions, notably, Jacobs and Lawvere-Ehrhard comprehension categories. First of all, we identify the structural principles that distinguish the latter from the former. In particular, we show that Lawvere-Ehrhard comprehension categories correspond to dependent type theories with type morphisms and a unit type, such that terms of type $A$ are in one-to-one correspondence with type morphisms from the unit type into $A$. This highlights that the difference between these two structures lies not in how they handle types and context extension, but rather in how they treat type morphisms: in Jacobs comprehension categories, they are an independent piece of data, whereas in Lawvere-Ehrhard comprehension categories, they are intrinsically connected to terms, which they completely determine.

Building on this characterization, the core contribution of the paper is the development of free constructions relating these two kinds of comprehension categories with each other and with plain fibrations. These free constructions capture, in a principled categorical way, the syntactic structures that differentiate these models from one another. At the same time, they provide a modular method for extending models with new features and for constructing a wide variety of new, free examples.

\subsubsection*{Outline} 
In \zcref{sect:prelim} we recall basic concepts and results about (Grothendieck) fibrations. 
\zcref{sect:jcomp} introduces the main characters of this paper, that is, (Jacobs) comprehension categories and Lawvere-Ehrhard comprehension categories. 
We first provide a characterization of those Jacobs comprehension categories that are actually Lawvere-Ehrhard  by comparing a fibration of terms with a fibration of type morphisms from the unit type. 
Then, we show how the two structures have different behaviours with respect to faithfulness and proof-irrelevance. 
In \zcref{sect:free-jcomp} 
we give a first universal construction 
describing the free (Jacobs) comprehension category over a fibration, 
Then, in \zcref{sect:JtoL} we give another universal construction turning any Jacobs comprehension category into a Lawvere-Ehrhard one. 
We do this in two steps: 
first we show how to freely add fibred terminal objects to a comprehension category, and 
then we build the free Lawvere-Ehrhard comprehension category over a Jacobs  comprehension category with fibred terminal objects. 
Note that, 
by composing all these three steps, we also get a construction of the free Lawvere-Ehrhard comprehension category over a fibration. 
Finally, \zcref{sect:conclu} summarizes our contribution and concludes the paper.

\section{Preliminaries on fibrations}\label{sect:prelim}

In this section we recall basic notions and results about (Grothendieck) fibrations, referring the reader to \cite{StreicherT:fibc} for more details. 
We assume familiarity with basic  concepts about (1-)categories and 2-categories, which can be found  in \eg, \cite{kelly2006review, johnson20212, riehl2017category}.
In the following we will write $\ctset$ for the category of sets and functions and $\Cattwo$ for the 2-category of categories, functors and natural transformations.
Given a category $\ct{C}$, we denote by $\Arr{\ct{C}}$ the category of arrows and commutative squares in $\ct{C}$. 
Similarly, given a 2-category $\cttwo{K}$, we denote by $\Arr{\cttwo{K}}$ the 2-category of 1-arrows, commutative squares and ``modifications" between such commutative squares in $\cttwo{K}$.
			
				Let $p:\ct{E}\to\ct{B}$ be a functor. An arrow $g:A\to B$ in $\ct{E}$ is \define{cartesian} if
				for every $g':A'\to B$ and $v:pA'\to pA$ such that $p(g')=p(g)\arcmp v$ there exist unique $h:A'\to A$ such that $g'=g\arcmp h$ and $p(h)=v$, as depicted in the following diagram
				\[\begin{tikzcd}
					{A'} \\
					A & B \\
					{pA'} \\
					pA & pB
					\arrow["h"', dashed, from=1-1, to=2-1]
					\arrow["{g'}", from=1-1, to=2-2]
					\arrow["g"', from=2-1, to=2-2]
					\arrow["v"', from=3-1, to=4-1]
					\arrow["{pg'}", from=3-1, to=4-2]
					\arrow["pg"', from=4-1, to=4-2]
				\end{tikzcd}\]
where the triangle below lies in $\ct{B}$ and the triangle above in $\ct{E}$.

				The functor $p$ is a  \define{fibration}  if, for every arrow $f:X\rightarrow Y$ in $\ct{B}$ and object $B$ in $\ct{E}$ over $Y$, there exist a  cartesian arrow $g:A\rightarrow B$ over $f$. 
				The arrow $g$ is a \define{cartesian lifting} of $f$ at $B$, $\ct{B}$ is the base category while $\ct{E}$ is the total category. 

We say that an arrow $f$ in $\ct{E}$ is \define{vertical} if $p(f)$ is an identity. 
It is easy to see that in a fibration  every arrow in the total category factorizes as a vertical arrow followed by a cartesian one. 
			
				A fibration is \define{cloven} if it is equipped with a choice, named \define{cleavage}, of cartesian liftings. 
That is, for every arrow $f:X\rightarrow Y$ in the base  and every object $B$ over $Y$, we have a cartesian arrow $\carar{f}{B}:\reind{f}{B}\to B$ over $f$. 
Note that, assuming the axiom of choice, every fibration is cloven. We will hence assume in the following that every fibration is endowed with a cleavage.
Another important fact to notice is that in general a cleavage does not preserve composition and identities, that is, 
$\carar{g}{\reind{f}B}\carar{f}{B}\neq \carar{(f\arcmp g)}{B}$ and $\carar{(\id{Y})}{B}\neq \id{B}$.
A cloven fibration where these equalities holds is called \define{split}. 

As already mentioned, fibrations provide a compact representation of families of categories. 
Indeed, given a fibration $p : \ct{E}\to\ct{B}$, for every object $X$ in $\ct{B}$ we can consider a category, dubbed \define{fibre} over $X$ and denoted by $\fibre{\ct{E}}{X}$, consisting of objects of $\ct{E}$ over $X$ and vertical arrows between them. 
Moreover, every arrow $f:X\to Y$ in $\ct{B}$ induces a \define{reindexing functor} $\reind{f} : \fibre{\ct{E}}{Y}\to\fibre{\ct{E}}{X}$. 
These data together give rise to a pseudofunctor $\ct{B}\op\to\Cattwo$, which is actually strict precisely when $p$ is split. 
This construction is known as the Grothendieck construction  and provides an equivalence between (cloven) fibrations and indexed categories. 
A fibration where the fibres are discrete categories, i.e., sets, is called \define{discrete}. 
These fibrations are necessarily cloven and split and correspond via the Grothendieck construction to standard $\ctset$-valued presheaves. 

			\begin{exa}\label{ex:cod}
				The codomain functor $\cod:\Arr{\ct{C}}\rightarrow \ct{C}$ is a fibration if and only if $\ct{C}$ has pullbacks, in fact an arrow in $\Arr{\ct{C}}$ is cartesian if and only if it is a pullback in $\ct{C}$. Hence a cartesian lifting of $f:A\to B$ at $g:C\to B$ is their pullback square.
			\end{exa}
			
			\begin{exa}\label{ex:typeth}
				Every dependent type theory induces a syntactic fibration $p : \ct{E}\to\ct{B}$ defined below.
				Objects in the base category $\ct{B}$ are contexts $\Gamma$, and arrows $\Gamma\to \Delta$, where $\Delta = y_1:\tau_1, ..., y_n:\tau_n$, is a $n$-tuple of terms $(M_1,...,M_n)$ satisfying $\Gamma \vdash M_i:\tau_i[M_1/y_1,..., M_{i-1}/y_{i-1}]$. These terms are to be interpreted as substitutions, and their composition is then the composition of substitutions. The objects of the total category $\ct{E}$ are type judgements of the form $\Gamma\vdash \sigma: \type$. The arrows $(\Gamma\vdash \sigma:\type)\to (\Delta\vdash \tau:\type)$ are pairs $(\vec{M},N)$ with $\vec{M}:\Gamma\to \Delta$ arrow in $\ct{B}$ and $N$ a term satisfying $\Gamma, x:\sigma\vdash N:\tau [\vec{M}/\vec{y}]$. Then the projection on the first component is a fibration.
				A cartesian lifting of an arrow $\vec{M}$ at a type judgement $\Delta\vdash \tau:\type$ is $(\vec{M},x)$ with $\Gamma, x:\tau[\vec{M}/\vec{y}]\vdash x:\tau[\vec{M}/\vec{y}]$.
				
				Let us notice that we are able to perform context extension: given $\sigma$ over $\Gamma$, one can consider the extended context $\Gamma, x:\sigma$. There is also a canonical projection $\chi\sigma:\Gamma, x:\sigma\to\Gamma$ given by the $n$-tuple of variables. Moreover, given an arrow $(\Gamma\vdash \sigma:\type)\to (\Delta\vdash \tau:\type)$, one can consider the following square:
				\[\begin{tikzcd}\label{diag:pb}
					{\Gamma,x:\sigma} && \Gamma \\
					{\Delta, y:\tau} && \Delta
					\arrow["{\chi\sigma}", from=1-1, to=1-3]
					\arrow["{(\vec{M},N)}"', from=1-1, to=2-1]
					\arrow["{\vec{M}}", from=1-3, to=2-3]
					\arrow["{\chi\tau}"', from=2-1, to=2-3]
				\end{tikzcd}\]
				It is not hard to see that this is a pullback in $\ct{B}$.

			This will give rise in \zcref{ex:typeth2} to a full comprehension category. We could have taken a complementary choice, obtaining a discrete comprehension category. The only difference is in the definition of morphisms in the total category: in the latter case a morphism $(\Gamma\vdash \sigma:\type)\to (\Delta\vdash \tau:\type)$ is defined as a morphism $\vec{M}:\Gamma\to\Delta$ in $\ct{B}$ such that $\tau[\vec{M}/\vec{y}]$ is (definitionally) equal to $\sigma$.
		 
			\end{exa}

In the rest of the paper we will make extensive use of fibrations with fibred terminal objects so we recall here some basic facts about them. 
Consider a fibration $p : \ct{E}\to\ct{B}$. 
A fibred terminal object is 
an object $A$ in a fibre $\fibre{\ct{E}}{Y}$ such that, for any morphism $f:X\to Y$ of $\ct{B}$ and object $B\in \fibre{\ct{E}}{Y}$, there exist a unique arrow $!:B\to A$ over $f$.

It is easy to see that this definition is equivalent to requiring that every fibre has terminal object and, for every $f:X\to Y$ in $\ct{B}$, the reindexing functor $\reind{f}:\fibre{\ct{E}}{Y}\to\fibre{\ct{E}}{X}$ preserves it. 
Fibrations with fibred terminal objects admit a further equivalent characterization, which we will often use throughout the paper. 

\begin{rem}
  A fibration $p : \ct{E}\to\ct{B}$ has fibred terminal objects if and only if there is a functor 
  $\term{p} : \ct{B}\to\ct{E}$ that is a right adjoint right inverse of $p$. 
  The functor $\term{p}$ picks for every object $X$ in $\ct{B}$ a fibred terminal object over it and transposition along the adjunction $p\dashv\term{p}$ provides us with the unique arrow $!_f : A \to \term{p} X$ over $f : pA\to X$, for every object $A$. 
Moreover, the functor $\term{p}$, being right adjoint and right inverse, is also full and faithful. 
\end{rem}
			
			\begin{exa}\label{ex:term}
				Consider the fibration $\cod:\Arr{\ct{C}}\to\ct{C}$ of \zcref{ex:cod}. A terminal object in $\fibre{\Arr{\ct{C}}}{X}$ is given by $\id{X}$, which is clearly stable under pullback. Hence, $\cod$ has fibred terminal objects.
			\end{exa}

Fibrations can be organized into a 2-category $\Fib$. 
Given fibrations $p:\ct{E}\to\ct{B}$ and $q:\ct{E'}\to\ct{B'}$, a 1-cell $F:p\to q$, also called fibration morphism, consists of a pair of functors $(\ft{F}, \fb{F})$ such that the square
				\[\begin{tikzcd}
					\ct{E} \ar[r, "\ft{F}"] 
					\ar[d,"p"'] & 
					\ct{E'} \ar[d,"q"] \\ 
					\ct{B} \ar[r,"\fb{F}"] &
					\ct{B'} 
				\end{tikzcd}\] 
commutes and $\ft{F}$ preserves cartesian arrows. 
Given fibration morphisms $F,G : p \to q$, 
 a 2-cell $\alpha:F\nt G$ is a pair of natural transformations $\alpha=(\ft{\alpha}, \fb{\alpha})$ where $\ft{\alpha}:\ft{F}\nt \ft{G}$ and $\fb{\alpha}:\fb{F}\nt\fb{G}$ and, moreover, $q\ft\alpha = \fb\alpha p$, i.e. the following diagram of 1- and 2-cells commutes 
				\[\begin{tikzcd}
					{\ct{E}} & {\ct{E'}} \\
					{\ct{B}} & {\ct{B'}}
					\arrow[""{name=0, anchor=center, inner sep=0}, "{{\ft{G}}}"{description}, bend right, from=1-1, to=1-2]
					\arrow[""{name=1, anchor=center, inner sep=0}, "{{\ft{F}}}"{description}, bend left, from=1-1, to=1-2]
					\arrow["p"', from=1-1, to=2-1]
					\arrow["q", from=1-2, to=2-2]
					\arrow[""{name=2, anchor=center, inner sep=0}, "{{\fb{F}}}"{description}, bend left, from=2-1, to=2-2]
					\arrow[""{name=3, anchor=center, inner sep=0}, "{{\fb{G}}}"{description}, bend right, from=2-1, to=2-2]
					\arrow["{\ft{\alpha}}", shorten <=3pt, shorten >=3pt, Rightarrow, from=1, to=0]
					\arrow["{\fb{\alpha}}", shorten <=3pt, shorten >=3pt, Rightarrow, from=2, to=3]
				\end{tikzcd}\]
Note that $\Fib$ is a 2-full 2-subcategory of $\Cattwo^2$. 
Given a category $\ct{B}$ we will denote by $\Fib(\ct{B})$ the 2-subcategory of $\Fib$ 
where objects are fibrations with base $\ct{B}$, 
1-cells are those $F$ such that $\fb{F} = \Id{\ct{B}}$, and 
2-cells are thos $\alpha$ such that $\fb{\alpha}_X = \id{X}$.

\section{Comprehension structures}\label{sect:jcomp}
	
In this section, after recalling the definition of (Jacobs) comprehension categories (JCC) and Lawvere-Ehrhard comprehension categories (LECC), 
we will provide a detailed comparison between them, highlighting the key features that distinguish the latter from the former ones. 
More precisely, if JCCs can be regarded as dependent type theories with type morphisms, 
we will show that LECCs are dependent type theories with type morphisms and an inhabited unit type and such that terms are completely determined by type morphisms from the unit type. 
We will also show that, when the considered fibrations are faithful, that is, type morphisms are actually a preorder between types over the same context, 
LECCs become proof irrelevant, i.e., every type has at most one inhabitant, while this does not happen for JCCs. 

Let us start by recalling the definitions of comprehension categories along with some examples. 

			\begin{defi}[{\cite[Def.~4.1]{JACOBS1993169}}]\label{def:compcat}
				A \define{Jacobs comprehension category} (JCC) is a fibration $p:\ct{E}\to\ct{B}$ together with a functor $\jcomp{p}:\ct{E}\to\Arr{\ct{B}}$ such that $\cod\arcmp \jcomp{p} = p$ and that $\jcomp{p}$ preserves cartesian arrows, i.e. $f$ cartesian in $\ct{E}$ implies $\jcomp{p} f$ is a pullback in $\ct{B}$. The functor $\jcomp{p}$ is called \define{comprehension functor}.
				\[\begin{tikzcd}
					{\ct{E}} & {\Arr{\ct{B}}} \\
					{\ct{B}}
					\arrow["{\jcomp{p}}", from=1-1, to=1-2]
					\arrow["p"', from=1-1, to=2-1]
					\arrow["\cod", from=1-2, to=2-1]
				\end{tikzcd}\]
				A comprehension category is \define{full} if its comprehension functor $\jcomp{}$ is full and faithful.
			\end{defi}

Given a JCC $p : \ct{E} \to \ct{B}$, 
we will denote by $\comp{p}$ the composition $\dom\arcmp\jcomp{p}$. For simplicity we will omit the fibration index whenever it will be clear from the context. 
			
			\begin{exa}\label{ex:cod-gen}
				The fibration $\cod$ of \zcref{ex:cod}, together with the identity $\Id{\Arr{\ct{B}}}$, is trivially a full comprehension category. A generalization to this is given by taking a family of arrows closed under pullback and considering the full subcategory of $\Arr{\ct{B}}$ on these morphisms. A particular case of this is obtained when the family of arrows consists of all monos. In this case the subcategory corresponds to the category of subobjects of $\ct{B}$.
			\end{exa}
			
			\begin{exa}\label{ex:typeth2}
				Consider the syntactic fibration defined in \zcref{ex:typeth}. The functor $\jcomp{}:\ct{E}\to\Arr{\ct{B}}$ is described by $\jcomp{}(\Gamma\vdash \sigma :\type)\colon =(\Gamma, x:\sigma)\to \Gamma$, so it sends a judgement of type in a context to the projection from the extended context to the old one. Explicitly, it is the list of variables of $\Gamma$, that is a list of terms in the extended context. With this definition we have a full comprehension category.

				We can also consider the discrete version, which is a comprehension category as well: cartesian morphisms correspond to those in the full case, hence they are still mapped to pullbacks.
			\end{exa}
			
The last example shows how every dependent type theory gives rise to a comprehension category. 
As already mentioned in the introduction, 
the connection is even tighter as one can build a type theory from any comprehension category, provided that one considers type morphisms as well. 
However, while contexts, substitutions, types and type morphisms can be easily extracted from a comprehension category, 
terms are less evident. 
Hence, it is useful to recall an equivalent presentation of comprehension categories, namely, generalized categories with families \cite[Def. 3.15]{emmenegger}. This has the advantage of being closer to the syntax by explicitly making terms as part of the structure. 

			\begin{defi}
				A \define{generalized category with families} (gcwf) consists of the data in the following diagram:
				\[\begin{tikzcd}
					{\tcat{\ct{E}}} && {\ct{E}} \\
					& {\ct{B}}
					\arrow[""{name=0, anchor=center, inner sep=0}, "\tsigma"', bend right=3ex, from=1-1, to=1-3]
					\arrow["{\tfib{p}}"', from=1-1, to=2-2]
					\arrow[""{name=1, anchor=center, inner sep=0}, "\Delta"', bend right=3ex, from=1-3, to=1-1]
					\arrow["p", from=1-3, to=2-2]
					\arrow["\dashv"{anchor=center, rotate=90}, draw=none, from=0, to=1]
				\end{tikzcd}\]
				
				where $p$, $\tfib{p}$ are fibrations and $\tsigma$ is a morphism of fibrations. The adjunction is such that both the components of unit and counit are cartesian with respect to $\tfib{p}$ and $p$ respectively. Note that $\Delta$ is not required to make the triangle involving $p$, $\tfib{p}$ to commute. If we ask that $p$, $\tfib{p}$ are both discrete, we get the definition of a \define{category with families} (cwf).
			\end{defi}

Intuitively, 
$\ct{B}$ is the category of contexts and substitutions, 
$p$ is the fibration of types and 
$\tfib{p}$ is the fibration of terms. 
The functor $\Sigma$ maps each term to its type, while the functor $\Delta$ 
maps a type to the generic element (i.e., the variable) of that type in the extended context, thus modelling the following rule: 
\[\inferrule*{\Gamma\vdash A\ \type}{\Gamma, \var{A}:A\vdash \var{A}:A}\]
The fact that $\Sigma$ is a morphism of fibrations ensures that 
every term should live in the same context of its type, and that given a substitution $\sigma:\Gamma\to \Delta$ and a term $t$ of type $A$ in context $\Delta$ we have that $t[\sigma]$ is of type $A[\sigma]$.
On the other hand, $\Delta$ cannot be in general a morphism of fibrations as the generic element of type $A$ lives in a context different from the one of $A$. 
This intuitive reading of a gcwf is justified by the following example. 

\begin{exa} 
\label{ex:gcwfsyntax} 
				Consider the syntactic fibration $p$ of \zcref{ex:typeth2}. We get a gcwf in which $\tcat{\ct{E}}$ is the category whose objects are judgements of the form $\Gamma \vdash t:A$ (i.e. terms in context), and morphisms $\Gamma \vdash t:A \to \Delta \vdash s:B$ are pairs $(\sigma, r)$ with $\sigma:\Gamma\to\Delta$ and $\Gamma, x:A\vdash r:B$ such that $r[t/x]=s[\sigma]$. The fibration $\tfib{p}$ sends $\Gamma \vdash t:A$ to $\Gamma$. 
				
				The functor $\tsigma$ returns the type of a term, i.e. $\tsigma(\Gamma \vdash t:A)= \Gamma \vdash A:\type$. Finally, the functor $\Delta$ assigns to each type the variable of that type in the extended context, following a rule of the form 
				\[\inferrule{\Gamma\vdash A:\type}{\Gamma, \var{A}:A\vdash \var{A}:A}\]
			\end{exa}

Let us notice that, in the example above, terms of type $A$ in context $\Gamma$ correspond to sections of the substitution $\jcomp{}A : \comp{}A \to \Gamma$, that is, of the comprehension of $A$. 
Indeed, such a section is by definition 
a list of terms in context $\Gamma$ such that its postcomposition with $\jcomp{}A$ is the identity. 
Hence, this list must consist of all the variables from $\Gamma$, plus a term of type $A$ (again in context $\Gamma$). 

This correspondence is the core of the construction of the gcwf associated with a comprehension category \cite{emmenegger}. 
Given a JCC $p:\ct{E}\to\ct{B}$, we define $\tfib{p}:\tcat{\ct{E}}\to \ct{B}$ in the following way:
			an object over $X$ is a pair of $A\in\fibre{\ct{E}}{X}$ together with a section $t:X\to \comp{}A$ of $\jcomp{}A$. A morphism $(A,t)\to(B,s)$ is then a morphism $f:A\to B$ over $g:X\to Y$ such that $\comp{}f\arcmp t=s\arcmp g$. Cartesian liftings with respect to $\tfib{p}$ are easily constructed considering cartesian liftings with respect to $p$ and the fact that $\jcomp{}$ maps cartesian morphisms to pullbacks.
			
			The action of $\tsigma$ is obvious: it maps $(A,t)$ to $A$. Furthermore, its fibres are discrete: given $f:(A,t)\to (A,s)$ over $\id{A}$, we have that $t=\comp{}\id{A}\arcmp t= s\arcmp \id{X}=s$. The action of $\Delta$ is less transparent: given $A$ over $X$, one can consider its reindex along $\jcomp{}A$, namely $\carar{(\jcomp{}A)}{A}:\reind{A}\to A$. Its image under $\jcomp{}$ is a pullback since it is cartesian, so one can define $\gel{A}:\comp{}A\to \comp{}\reind{A}$ by its universal property, as depicted below:
			\[\begin{tikzcd}
				{\comp{}A} \\
				& {\comp{}\reind{A}} && {\comp{}A} \\
				& {\comp{}A} && X
				\arrow["{{\gel{A}}}", dashed, from=1-1, to=2-2]
				\arrow["{{\id{\comp{}A}}}", bend left=3ex, from=1-1, to=2-4]
				\arrow["{{\id{\comp{}A}}}"', bend right, from=1-1, to=3-2]
				\arrow["{{\jcomp{}\reind{A}}}", from=2-2, to=2-4]
				\arrow["{{\comp{}\carar{(\jcomp{}A)}{A}}}", from=2-2, to=3-2]
				\arrow["{{\jcomp{}A}}", from=2-4, to=3-4]
				\arrow["{{\jcomp{}A}}"', from=3-2, to=3-4]
			\end{tikzcd}\]
			Since $\gel{A}$ is a section of $\jcomp{}\reind{A}$ by construction, we can define $\Delta(A)$ as the pair $(\reind{A}, \gel{A})$.
		
This construction formally justifies the type-theoretic reading of comprehension categories. 
Indeed, we will call \define{types} the objects of the total category, and \define{terms of type $A$} the sections of the comprehension $\jcomp{}A$. Then we will call \define{type morphisms} vertical arrows in $\ct{E}$, and a type morphism is said \define{global} if its domain is terminal in the fibre.
			Moreover one can consider a cartesian lifting $\weakar{A}{B}:\weak{A}{B}\to B$ of $\jcomp{}A$ at $B$. We call $\weak{A}{B}$ the \define{weakening of B along A}. Finally we call \define{generic element of type $A$} the unique arrow $\gel{A}$ given by the universal property of the pullback, as above.

We now recall the definition of Lawvere-Ehrhard comprehension categories.

			\begin{defi}[{\cite[Def.~5]{ehrhard1988}}]
				A \define{Lawvere-Ehrhard comprehension category} (LECC) is a fibration $p:\ct{E}\rightarrow\ct{B}$ together with two functors $\term{p}:\ct{B}\rightarrow\ct{E}$ and $\comp{p}:\ct{E}\rightarrow\ct{B}$ such that $\term{p}$ is a fibred terminal object functor (or equivalently $\term{p}$ is right adjoint and right inverse to $p$) and $\comp{p}$ is right adjoint to $\term{p}$.
				\[\begin{tikzcd}
					{\ct{E}} \\
					\\
					{\ct{B}}
					\arrow[""{name=0, anchor=center, inner sep=0}, "{\comp{p}}", bend left=10ex, from=1-1, to=3-1]
					\arrow[""{name=1, anchor=center, inner sep=0}, "p"', bend right=10ex, from=1-1, to=3-1]
					\arrow[""{name=2, anchor=center, inner sep=0}, "{\term{p}}"{description}, from=3-1, to=1-1]
					\arrow["\dashv"{anchor=center}, draw=none, from=1, to=2, pos=0.4]
					\arrow["\dashv"{anchor=center}, draw=none, from=2, to=0, pos=0.55]
				\end{tikzcd}\]
			\end{defi}
			For simplicity we will omit the fibration index whenever it will be clear from the context.
			
			Notice that LECCs are determined by adjointness, hence there is at most one (up to iso) Lawvere-Ehrhard comprehension structure over a given fibration. This allows us to say that a fibration (or a JCC) \emph{is} a LECC whenever it satisfies the appropriate condition.
						
			\begin{exa}
				Consider the fibration $\cod$ of \zcref{ex:cod}. It is a LECC: the fibred terminal object functor sends an object in the base to the identity, as described in \zcref{ex:term}, while the comprehension functor is $\dom$. It is easy to see that the adjunctions $\cod\dashv \term{p}\dashv \dom$ hold.
			\end{exa}
			
			\begin{exa}
				Consider the fibration $\mathcal{P}:\ct{S}\to\ctset$. Objects in $\ct{S}$ are pairs of a set $X$ and a subset $S\subseteq X$, and an arrow $(X,S)\to(Y,T)$ is a function $f:X\to Y$ such that $S\subseteq f^{-1}(T)$. This fibration is a LECC: the comprehension functor $\comp{\mathcal{P}}$ is obtained by taking the second component of the pair, and on an arrow $f:(X,S)\to(Y,T)$ it gives the restriction and corestriction of $f$ to, respectively, $S$ and $T$. 
Notice that this fibration is obtained by applying the Grothendieck construction \cite{jacobs1999categorical} to the powerset functor $\mathcal{P}:\ctset\op\to\Cat$ mapping a set to its powerset ordered by inclusion and regarded as a category. This implies the faithfulness of the fibration.
			\end{exa}

			\subsection{Comparing Jacobs and Lawvere-Ehrhard comprehension categories}
			\label{sect:lcomp-vs-jcomp}

In this section we show that Lawvere-Ehrhard comprehension categories are indeed a special type of (Jacobs) comprehension categories by describing a 2-functor from the 2-category of LECCs to the one of JCCs (\zcref{thm:l-to-j}). 
Moreover, we provide a characterization of the essential image of such a 2-functor (\zcref{thm:ess-im}), 
thus highlighting the key features that  distinguish Lawvere-Ehrhard from Jacobs comprehension categories. 
Finally, in \zcref{ex:ourfamsetp} we show that the JCC structure over a fibration is not unique. 
This points out an important difference between LECCs and JCCs: 
the former, being defined by adjunctions, is a property of the fibration, while the latter consists of structure on top of it. 

We start by introducing the 2-categories of Jacobs comprehension categories and Lawvere-Ehrhard comprehension categories. 

For Jacobs comprehension categories, 1-cells are fibration morphisms preserving comprehensions up to isomorphism 
and 2-cells are those of fibrations which are compatible with the specified isomorphisms of 1-cells. 

More precisely,  given two JCCs $p:\ct{E}\to\ct{B}$ and $q:\ct{E'}\to \ct{B'}$, a \define{morphism of Jacobs comprehension categories} from $p$ to $q$ consists of a fibration morphism $F:p\to q$ together with a natural isomorphism $\alpha:(\jcomp{q}\arcmp \ft{F})\nt (\Arr{(\fb{F})}\arcmp \jcomp{p})$ such that $\lwhisk{\cod}\alpha=\idtwo{\fb{F}\arcmp p}$. Here $\idtwo{\fb{F}\arcmp p}$ denotes the identity natural transformation on $\fb{F}\arcmp p$.

			Given two morphisms of JCCs $F:p\to q$ and $G:q\to s$ together with $\alpha$ and $\beta$ respectively, their composition is given by $G\arcmp F$ together with $\beta * \alpha\colon=(\Arr{(\JFun{G})}\lwhisk{\alpha})\arcmp (\beta\rwhisk{\ft{F}})$.
				\[\begin{tikzcd}
				& {\Arr{\ct{B}}} && {\Arr{\ct{B'}}} && {\Arr{\ct{B''}}} \\
				{\ct{E}} && {\ct{E'}} && {\ct{E''}} \\
				{\ct{B}} && {\ct{B'}} && {\ct{B''}}
				\arrow[""{name=0, anchor=center, inner sep=0}, "{{\Arr{(\fb{F})}}}", from=1-2, to=1-4]
				\arrow["\cod"{pos=0.7}, bend left, from=1-2, to=3-1]
				\arrow[""{name=1, anchor=center, inner sep=0}, "{{\Arr{(\fb{G})}}}", from=1-4, to=1-6]
				\arrow["\cod"{pos=0.7}, bend left, from=1-4, to=3-3]
				\arrow["\cod"{pos=0.7}, bend left, from=1-6, to=3-5]
				\arrow["{{\jcomp{p}}}"', from=2-1, to=1-2]
				\arrow["{{\ft{F}}}", from=2-1, to=2-3, crossing over, pos=0.65]
				\arrow["p"', from=2-1, to=3-1]
				\arrow[""{name=2, anchor=center, inner sep=0}, "{{\jcomp{q}}}"', from=2-3, to=1-4]
				\arrow["{{\ft{G}}}", from=2-3, to=2-5, crossing over, pos=0.65]
				\arrow["q", from=2-3, to=3-3]
				\arrow[""{name=3, anchor=center, inner sep=0}, "{{\jcomp{s}}}"', from=2-5, to=1-6]
				\arrow["s", from=2-5, to=3-5]
				\arrow["{{\fb{F}}}"', from=3-1, to=3-3]
				\arrow["{{\fb{G}}}"', from=3-3, to=3-5]
				\arrow["\alpha", bend left, shorten <=6pt, shorten >=6pt, Rightarrow, from=2, to=0]
				\arrow["\beta", bend left, shorten <=6pt, shorten >=6pt, Rightarrow, from=3, to=1]
			\end{tikzcd}\]
			
Let $F,G:p\to q$ together with $\alpha,\beta$ be morphisms of JCCs. A 2-cell of Jacobs comprehension categories $(F,\alpha)\nt(G,\beta)$ is a 2-cell $\phi:F\nt G$ in $\Fib$ such that $\lwhisk{\jcomp{q}}\ft{\phi}\arcmp\alpha^{-1}=\beta^{-1}\arcmp\Arr{(\fb{\phi})}\rwhisk{\jcomp{p}}$.
			\[\begin{tikzcd}[row sep=5ex,column sep=7ex]
				{\ct{E}} & {\ct{E'}} & {\Arr{\ct{B}}} & {\Arr{\ct{B'}}} & {\Arr{\ct{B}}} & {\Arr{\ct{B'}}} \\
				{\ct{B}} & {\ct{B'}} & {\ct{E}} & {\ct{E'}} & {\ct{E}} & {\ct{E'}}
				\arrow[""{name=0, anchor=center, inner sep=0}, "{\ft{G}}"{description}, bend right, from=1-1, to=1-2]
				\arrow[""{name=1, anchor=center, inner sep=0}, "{\ft{F}}"{description}, bend left, from=1-1, to=1-2]
				\arrow["p"', from=1-1, to=2-1]
				\arrow["q", from=1-2, to=2-2]
				\arrow[""{name=2, anchor=center, inner sep=0}, "{\Arr{(\fb{F})}}"{description}, bend left, from=1-3, to=1-4]
				\arrow[""{name=3, anchor=center, inner sep=0}, "{\Arr{(\fb{G})}}"{description}, bend right, from=1-5, to=1-6]
				\arrow[""{name=4, anchor=center, inner sep=0}, "{\Arr{(\fb{F})}}"{description}, bend left, from=1-5, to=1-6]
				\arrow[""{name=5, anchor=center, inner sep=0}, "{\fb{F}}"{description}, bend left, from=2-1, to=2-2]
				\arrow[""{name=6, anchor=center, inner sep=0}, "{\fb{G}}"{description}, bend right, from=2-1, to=2-2]
				\arrow["{\jcomp{p}}", from=2-3, to=1-3]
				\arrow[""{name=7, anchor=center, inner sep=0}, "{\ft{G}}"{description}, bend right, from=2-3, to=2-4]
				\arrow[""{name=8, anchor=center, inner sep=0}, "{\ft{F}}"{description}, bend left, from=2-3, to=2-4]
				\arrow["{\jcomp{q}}"', from=2-4, to=1-4]
				\arrow["{\jcomp{p}}", from=2-5, to=1-5]
				\arrow[""{name=9, anchor=center, inner sep=0}, "{\ft{G}}"{description}, bend right, from=2-5, to=2-6]
				\arrow["{\jcomp{q}}"', from=2-6, to=1-6]
				\arrow["\fb{\phi}", shorten <=3pt, shorten >=3pt, Rightarrow, from=1, to=0]
				\arrow["{\alpha^{-1}}", shorten <=4pt, shorten >=4pt, Rightarrow, from=2, to=8]
				\arrow["{\Arr{(\fb{\phi})}}", shorten <=3pt, shorten >=3pt, Rightarrow, from=4, to=3]
				\arrow["{\beta^{-1}}", shorten <=4pt, shorten >=4pt, Rightarrow, from=3, to=9]
				\arrow["\ft{\phi}", shorten <=3pt, shorten >=3pt, Rightarrow, from=5, to=6]
				\arrow["\fb{\phi}", shorten <=3pt, shorten >=3pt, Rightarrow, from=8, to=7]
			\end{tikzcd}\]	
			We denote by $\JComp$ the 2-category of JCCs with 1-cells and 2-cells as described above. 
			
For Lawvere-Ehrhard comprehension categories  definitions are similar but with some simplifications due to the presence of adjunctions. 
Roughly, 1-cells are fibration morphisms preserving fibred terminal objects and comprehensions up to iso and 2-cells are the same as those of fibrations. 
More precisely, 
given Lawvere-Ehrhard comprehension categories  $p:\ct{E}\rightarrow\ct{B}$ and $q:\ct{E'}\rightarrow\ct{B'}$, 
a \define{morphism of Lawvere-Ehrhard comprehension categories} (LE-morphism for short) from $p$ to $q$ is a fibration morphism $F:p\to q$ such that the natural isomorphism $\theta:\ft{F}\arcmp \term{p}\nt \term{q}\arcmp \fb{F}$ determined as the mate of the identity natural transformation $\idtwo{\fb{F}\arcmp p}$ is invertible and its inverse's mate $\mate{(\theta^{-1})}:\fb{F}\arcmp \comp{p}\nt \comp{q}\arcmp \ft{F}$ is again invertible.
Note that, differently from Jacobs comprehension categories, the natural isomorphisms witnessing the preservation of fibred terminal objects and comprehensions are uniquely determined by the rest of the data. 
Observe also that, 
given LE-morphisms $F:p\to q$ and $G:q\to s$, their composition is again a LE-morphism. 
Indeed, 
given the natural isomorphisms $\theta:\ft{F}\arcmp \term{p}\nt \term{q}\arcmp \fb{F}$ and $\sigma:\ft{G}\arcmp \term{q}\nt \term{s}\arcmp \fb{G}$ associated with the two LE-morphisms, 
the natural morphism  $(\sigma\rwhisk{\fb{F}})\arcmp (\lwhisk{\ft{G}}\theta)$ is invertible. Furthermore, the mate of its inverse is invertible as well, showing that $G\arcmp F$ is a LE-morphism.
				
			\[\begin{tikzcd}
				{\ct{E}} & {\ct{E'}} & {\ct{E''}} \\
				{\ct{B}} & {\ct{B'}} & {\ct{B''}}
				\arrow[""{name=0, anchor=center, inner sep=0}, "{\ft{F}}", from=1-1, to=1-2]
				\arrow[""{name=1, anchor=center, inner sep=0}, "{\ft{G}}", from=1-2, to=1-3]
				\arrow["{\term{p}}", from=2-1, to=1-1]
				\arrow["{\fb{F}}"', from=2-1, to=2-2]
				\arrow[""{name=2, anchor=center, inner sep=0}, "{\term{q}}"', from=2-2, to=1-2]
				\arrow["{\fb{G}}"', from=2-2, to=2-3]
				\arrow[""{name=3, anchor=center, inner sep=0}, "{\term{s}}"', from=2-3, to=1-3]
				\arrow["\theta"', bend right, shorten <=4pt, shorten >=4pt, Rightarrow, from=0, to=2]
				\arrow["\sigma"', bend right, shorten <=4pt, shorten >=4pt, Rightarrow, from=1, to=3]
			\end{tikzcd}\]
			Finally, the 2-category $\LComp$ is the 2-full 2-subcategory of $\Fib$ spanned by Lawvere-Ehrhard comprehension categories and LE-morphisms between them. 
			
Jacobs in \cite[Def.~4.12]{JACOBS1993169} proves that a LECC gives rise to a JCC by setting $\jcomp{}A=p\epsilon_A$, with $\epsilon$ the counit of the adjunction $\term{}\dashv\comp{}$. 
The following theorem extends Jacobs' result by showing that this assignment gives rise to a 2-functor. 

			\begin{thm} \label{thm:l-to-j}
				Let $p:\ct{E}\rightarrow \ct{B}$ be a LECC. Then $(p,\jcomp{p})$ is a JCC, where $\jcomp{p}:\ct{E}\rightarrow\Arr{\ct{B}}$ is defined by $\jcomp{p}(A)\colon= p\coun{p}_A$, with $\coun{p}$ the counit of the comprehension-terminal adjunction. Furthermore this assignment extends to a 2-functor $\LtoJ:\LComp\to\JComp$.
			\end{thm}
			
			\begin{proof}
				Of course $\cod\arcmp\jcomp{p}=p$. So we only need to verify that if $f:A\to B$ is cartesian, then $\jcomp{}f$ is a pullback in $\ct{B}$. Consider a pair of arrows $g:Z\to\comp{}B$ and $h:Z\to X$, where $X=pA$, such that $\jcomp{}B\arcmp g=pf\arcmp h$. The transpose $\coun{}_B\arcmp\term{}g:\term{}Z\to B$ of $g$ is over $\jcomp{}B\arcmp g=pf\arcmp h$, so by cartesianity of $f$ there is a unique $s:\term{}Z\to A$ over $h$ such that $f\arcmp s=\coun{}_B\arcmp \term{}g$. This yields a unique arrow $s^{\#}:Z\to \comp{}A$ by taking the transpose of $s$. Furthermore we have the following
				\[\comp{}f\arcmp s^{\#}=\comp{}(f\arcmp s)\arcmp \un{}_Z=\comp{}(\coun{}_B\arcmp \term{}g)\arcmp\un{}_Z=g\] 
				\[\jcomp{}A\arcmp s^{\#}=p(\coun{}_A\arcmp\term{}\comp{}s\arcmp\term{}\un{}_Z)=p(s\arcmp\coun{}_{\term{}Z}\arcmp\term{}\un{}_Z)=p(s)=h\]
				For the first equation we used the characterization of transposes via unit and counit, and for the second we used also the naturality of the counit on $s$.
				\[\begin{tikzcd}
					{\term{}\comp{}\term{}Z} & {\term{}Z} & {\term{}\comp{}B} \\
					& {\term{}\comp{}A} & A & B \\
					& Z \\
					&& {\comp{}A} & X \\
					&& {\comp{}B} & Y
					\arrow["{\coun{}_{\term{}Z}}", from=1-1, to=1-2]
					\arrow["{\term{}\comp{}s}"', from=1-1, to=2-2]
					\arrow["{\term{}g}", from=1-2, to=1-3]
					\arrow["s"', dashed, from=1-2, to=2-3]
					\arrow["{\coun{}_B}", from=1-3, to=2-4]
					\arrow["{\coun{}_A}"', from=2-2, to=2-3]
					\arrow["f"', from=2-3, to=2-4]
					\arrow["{s^{\#}}"{description}, dashed, from=3-2, to=4-3]
					\arrow["h", bend left, from=3-2, to=4-4]
					\arrow["g"', bend right, from=3-2, to=5-3]
					\arrow["{\jcomp{}A}", from=4-3, to=4-4]
					\arrow["{\comp{}f}"', from=4-3, to=5-3]
					\arrow["pf", from=4-4, to=5-4]
					\arrow["{\jcomp{}B}"', from=5-3, to=5-4]
				\end{tikzcd}\]
				
				Now, let $F:p\to q$ together with $\theta:\ft{F}\arcmp \term{p}\nt \term{q}\arcmp \fb{F}$ be a 1-cell in $\LComp$. In order to show that $F$ is a morphism in $\JComp$ as well we only need to define a natural isomorphism $\alpha:(\Arr{(\fb{F})}\arcmp \jcomp{p})\nt (\jcomp{q}\arcmp \ft{F})$ such that $\lwhisk{\cod}\alpha=\idtwo{\fb{F}p}$. Let $A$ be an object in $\ct{E}$ over $X$ and consider the square 
				\[\begin{tikzcd}
					{\term{q}\fb{F}\comp{p}A} & {\ft{F}\term{p}\comp{p}A} \\
					{\term{q}\comp{q}\ft{F}A} & {\ft{F}A}
					\arrow["{\theta^{-1}_{\comp{p}A}}", from=1-1, to=1-2]
					\arrow["{\term{q}\mate{(\theta^{-1})}_A}"', from=1-1, to=2-1]
					\arrow["{\ft{F}\coun{p}_A}", from=1-2, to=2-2]
					\arrow["{\coun{q}_{\ft{F}A}}"', from=2-1, to=2-2]
				\end{tikzcd}\]
				It commutes as a consequence of the definition of mate. Applying $q$ to this gives a commutative square in $\ct{B'}$ whose top side is the identity, since $\lwhisk{\cod}\theta=\idtwo{\Id{\ct{B'}}}$. Then we can set $\alpha_A\colon=(\mate{(\theta^{-1})}_A,\id{\fb{F}X})$. It is a natural iso since both its components are invertible.
				\[\begin{tikzcd}
					{\fb{F}\comp{p}A} & {\fb{F}\comp{p}A} \\
					{\comp{q}\ft{F}A} & {\fb{F}X}
					\arrow["{\id{\fb{F}\comp{p}A}}", from=1-1, to=1-2]
					\arrow["{\mate{(\theta^{-1})}_A}"', from=1-1, to=2-1]
					\arrow["{\fb{F}p\coun{p}_A}", from=1-2, to=2-2]
					\arrow["{q\coun{q}_{\ft{F}A}}"', from=2-1, to=2-2]
				\end{tikzcd}\]
				
				Finally its action on the 2-cells is given by the identity. 
				In fact, consider $\phi:F\to G$ a 2-cell in $\LComp$. Then $\sigma\arcmp\ft{\phi}\rwhisk{\term{p}}=\lwhisk{\term{q}}\fb{\phi}\arcmp \theta$ and $\mate{(\sigma^{-1})}\arcmp\fb{\phi}\rwhisk{\comp{p}}=\lwhisk{\comp{q}}\ft{\phi}\arcmp\mate{(\theta^{-1})}$, where $\theta$ and $\sigma$ are the isomorphisms depicted below. Diagrammatically, the following squares of natural transformations commute:
				\[\begin{tikzcd}
					{\ft{F}\term{p}} & {\term{q}\fb{F}} && {\comp{q}\ft{F}} & {\fb{F}\comp{p}} \\
					{\ft{G}\term{p}} & {\term{q}\fb{G}} && {\comp{q}\ft{G}} & {\fb{G}\comp{p}}
					\arrow["\theta", Rightarrow, from=1-1, to=1-2]
					\arrow["{\ft{\phi}\rwhisk{\term{p}}}"', Rightarrow, from=1-1, to=2-1]
					\arrow["{\lwhisk{\term{q}}\fb{\phi}}", Rightarrow, from=1-2, to=2-2]
					\arrow["{\lwhisk{\comp{q}}\ft{\phi}}"', Rightarrow, from=1-4, to=2-4]
					\arrow["{\mate{(\theta^{-1})}}"', Rightarrow, from=1-5, to=1-4]
					\arrow["{\fb{\phi}\rwhisk{\comp{p}}}", Rightarrow, from=1-5, to=2-5]
					\arrow["\sigma"', Rightarrow, from=2-1, to=2-2]
					\arrow["{\mate{(\sigma^{-1})}}", Rightarrow, from=2-5, to=2-4]
				\end{tikzcd}\]
				This implies that 2-cells of $\LComp$ preserve comprehension and terminal objects. It is not hard to see that they satisfy the coherence required.
			\end{proof}

Our next goal is to provide an ``intrinsic'' characterization of those comprehension categories which are actually Lawvere-Ehrhard, that is, 
to characterize the essential image of the 2-functor $\LtoJ$. 
We start by observing an interesting property of 
comprehension categories with fibred terminal objects:  
we can not only construct the fibration of terms $\tfib{p}$, but also a fibration of global type morphisms. 
Fix a comprehension category $p:\ct{E}\to\ct{B}$ with fibred terminal objects $\term{}$. 
The fibration $\gtfib{p}:\gtcat{\ct{E}}\to \ct{B}$ of global type morphisms is defined as follows. 
An object of $\gtcat{\ct{E}}$ over $X$ is a pair $(A,t)$ of an object $A$ over $X$ in $p$ and a vertical arrow $t:\term{}X\to A$, and 
a morphism from $(A,t)$ to $(B,s)$ over $g:X\to Y$ is a morphism $f:A\to B$ over $g$ such that $f\arcmp t=s\arcmp\term{}g$.
\[\begin{tikzcd}
	{\term{}X} & {\term{}Y} \\
	A & B
	\arrow["{\term{}g}", from=1-1, to=1-2]
	\arrow["t"', from=1-1, to=2-1]
	\arrow["s", from=1-2, to=2-2]
	\arrow["f"', from=2-1, to=2-2]
\end{tikzcd}\]
The fibration $\gtfib{p}:\gtcat{\ct{E}}\to \ct{B}$ obviously maps $(A,t)$ to $X$, and $f$ to $g$. Cartesian liftings are obtained straightforwardly using cartesian liftings of $p$.
As in the construction of the gcwf associated to a JCC, there is a morphism of fibrations $\gtsigma:\gtfib{p}\to p$ mapping $(A,t)$ to $A$ and morphisms to themselves. Also in this case the fibres of $\gtsigma$ are discrete: given $f:(A,t)\to(A,s)$ over $\id{A}$, we have that $t=\id{A}\arcmp t=s\arcmp\term{}\id{X}=s$.

We are going to prove in \zcref{thm:ess-im} that a comprehension category with fibred terminal objects is Lawvere-Ehrhard if and only if the associated fibrations of terms and global type morphisms are isomorphic. So, we start by showing that this condition is necessary.
To this end, it is useful to notice first another necessary condition for a comprehension category to be Lawvere-Ehrhard: the associated comprehension functor à la Jacobs must preserve fibred terminal objects, i.e., the comprehension of a fibred terminal must be an isomorphism. 
From a type-theoretic perspective, this means that in a LECC there is always a unique term of the unit type in any context. 

\begin{lem}\label{prop:comp-term-inv}
	Let $p:\ct{E}\to\ct{B}$ be a LECC, and $X$ an object in $\ct{B}$. Then the arrow $p\coun{}_{\term{}X}:\comp{}\term{}X\to X$ is invertible, with inverse given by $\un{}_X$. 
\end{lem}

\begin{proof}
	Since $\term{}$ is both a right adjoint and a section, it is full and faithful. This implies that $\eta$ is an iso (see \cite[Lemma ~4.5.13]{riehl2017category}). Triangular identities show that its inverse is given by $p\epsilon_\term{}$.
\end{proof}

Next, recall that in a Lawvere-Ehrhard comprehension category $p : \ct{E}\to\ct{B}$ 
transposition along the adjunction $\term{} \dashv \comp{}$ is a bijection between morphisms in $\ct{E}$ of the form $\term{}X\to A$ and morphisms in $\ct{B}$ of the form $X\to\comp{}A$. This bijection restricts to global type morphisms on one side and terms on the other.

			\begin{lem}\label{lem:prf-trm}
				Let $p:\ct{E}\to\ct{B}$ be a LECC. Then transposition restricts to a bijection $\Obj\tcat{\ct{E}}\cong\Obj\gtcat{\ct{E}}$.
			\end{lem}
			
			\begin{proof}
				Given $A$ in $\ct{E}$ over $X$ and a section $f:X\to \comp{}A$ of $\jcomp{}A:\comp{}A\to X$, its transpose is $f^{\#}=\coun{}_A\arcmp\term{}f:\term{}X\to A$. This is a global type morphism since it is vertical ($f$ is a section) and it is from a terminal object.
				\[\begin{tikzcd}
					{\term{}X} & {\term{}\comp{}A} & A \\
					\\
					X & {\comp{}A} & X
					\arrow["{\term{}f}"', from=1-1, to=1-2]
					\arrow["{f^{\#}}", bend left, from=1-1, to=1-3]
					\arrow["{\coun{}_A}"', from=1-2, to=1-3]
					\arrow["f"', from=3-1, to=3-2]
					\arrow["{\id{X}}", bend left, from=3-1, to=3-3]
					\arrow["{\jcomp{}A}"', from=3-2, to=3-3]
				\end{tikzcd}\]
				
				Conversely, if we transpose a global type morphism $g:\term{}Y\to B$, we get $g^{\#}=\comp{}g\arcmp\un{}_Y:Y\to\comp{}B$. Its postcomposition with $\jcomp{}B$ is the identity: consider the naturality square of the counit on $g$. Applying $p$ to it yields the equality $\jcomp{}B\arcmp\comp{}g=\jcomp{}{\term{}Y}$. We observed in \zcref{prop:comp-term-inv} that the unit $\un{}_Y$ is invertible, and that its inverse is given by $p\coun{}_{\term{}Y}=\jcomp{}{\term{}Y}$, thus proving that $g^{\#}$ is a section of $\jcomp{}B$.
				\[\begin{tikzcd}
					& {\term{}\comp{}\term{}Y} & {\term{}Y} \\
					& {\term{}\comp{}B} & B \\
					\\
					Y & {\comp{}\term{}Y} & Y \\
					& {\comp{}B} & Y
					\arrow["{\coun{}_{\term{}Y}}", from=1-2, to=1-3]
					\arrow["{\term{}\comp{}g}"', from=1-2, to=2-2]
					\arrow["g", from=1-3, to=2-3]
					\arrow["{\coun{}_B}"', from=2-2, to=2-3]
					\arrow["{\un{}_Y}"', from=4-1, to=4-2]
					\arrow["{\id{Y}}", bend left, from=4-1, to=4-3]
					\arrow["{\jcomp{}\term{}Y}"', from=4-2, to=4-3]
					\arrow["{\comp{}g}"', from=4-2, to=5-2]
					\arrow["{\id{Y}}", from=4-3, to=5-3]
					\arrow["{\jcomp{}B}"', from=5-2, to=5-3]
				\end{tikzcd}\]
			\end{proof}

Given a Lawvere-Ehrhard comprehension category $p : \ct{E}\to\ct{B}$, 
the previous lemma provides us 
with a bijective correspondence, given by transposition, between 
the objects of $\tcat{\ct{E}}$ and $\gtcat{\ct{E}}$, i.e. terms and global type morphisms. 
This bijection extends then to an isomorphism of fibrations.

			\begin{prop}\label{prop:prf-trm}
				Let $p:\ct{E}\to\ct{B}$ be a LECC. Then transposition induces an isomorphism $\ftransp:\gtsigma\to\tsigma$ in $\Fib(\ct{B})/p$.
			\end{prop}
			
			\begin{proof}
				A morphism $\ftransp{}:\gtsigma\to\tsigma$ in $\Fib(\ct{B})/p$ is just a functor $\ttransp:\gtcat{\ct{E}}\to\tcat{\ct{E}}$ which preserves cartesian arrows and makes the triangles below to commute.
				\[\begin{tikzcd}
					{\gtcat{\ct{E}}} &&&& {\tcat{\ct{E}}} \\
					&& {\ct{E}} \\
					\\
					&& {\ct{B}}
					\arrow["{\ttransp{}}", from=1-1, to=1-5]
					\arrow["\gtsigma"{description}, from=1-1, to=2-3]
					\arrow["{\gtfib{p}}"', from=1-1, to=4-3, bend right=3ex]
					\arrow["\tsigma"{description}, from=1-5, to=2-3]
					\arrow["{\tfib{p}}", from=1-5, to=4-3, bend left=3ex]
					\arrow["p"{description}, from=2-3, to=4-3]
				\end{tikzcd}\]
				
				This last condition is obvious by \zcref{lem:prf-trm} at the level of objects. For morphisms, it is sufficient to prove that given $f:A\to B$ in $\ct{E}$ we have that each one of the following two squares commutes if and only if the other does.
				\[\begin{tikzcd}
					X & {\comp{}A} && {\term{}X} & A \\
					Y & {\comp{}B} && {\term{}Y} & B
					\arrow["t", from=1-1, to=1-2]
					\arrow["g"', from=1-1, to=2-1]
					\arrow["{\comp{}f}", from=1-2, to=2-2]
					\arrow["{t^\#}", from=1-4, to=1-5]
					\arrow["{\term{}g}"', from=1-4, to=2-4]
					\arrow["f", from=1-5, to=2-5]
					\arrow["s"', from=2-1, to=2-2]
					\arrow["{s^\#}"', from=2-4, to=2-5]
				\end{tikzcd}\]
				This is true just by unfolding the definition of transposition using the unit and the counit of the adjunction $\term{}\dashv\comp{}$.
			\end{proof}

The next step is to show that the isomorphism between the fibrations of terms and global type morphisms is sufficient for a comprehension category to be Lawvere-Ehrhard. 
In \zcref{lem:term-pres} we show that, whenever a comprehension category has an inhabited unit type, 
we have a way to map a global type morphism into a term. 
Having inhabited unit types 
means that 
the comprehension category has fibred terminal objects and their comprehensions $\jcomp{}\term{}$ have a natural section $\eta$. 
In type-theoretic terms, this means that 
the following rules hold 
\begin{mathpar}
\inferrule{ \vdash \Gamma\ \ctx }{ \Gamma \vdash \unitty\ \type }
\and 
\inferrule{\vdash \Gamma \ \ctx}{\Gamma\vdash \unittr:\unitty}
\end{mathpar}
Then, \zcref{lem:ess-im} characterizes Lawvere-Ehrhard comprehension categories as those comprehension categories with inhabited unit types with a bijective correspondence between terms and global type morphisms. 
			
			\begin{lem}\label{lem:term-pres}
				Let $p:\ct{E}\to\ct{B}$ be a comprehension category with fibred terminal objects $\term{}$. Then there is a morphism $F:\gtsigma\to\tsigma$ in $\Fib(\ct{B})/p$ if and only if there for any $X$ in $\ct{B}$ there is a section $\eta_X:X\to\comp{}\term{}X$ of the comprehension $\jcomp{}\term{}X$ natural in $X$.
				Moreover, the components of $\eta$ are obtained as the second components of $F(\term{}, \id{\term{}})$.
			\end{lem}
			
			\begin{proof}
				Assume first that there is $F:\gtcat{\ct{E}}\to\tcat{\ct{E}}$ in $\Fib(\ct{B})/p$. Then a section $\un{}_X$ of $\jcomp{}\term{}X$ is easily obtained as the second component of $F(\term{}X, \id{\term{}X})$. The naturality is easily deduced from the commutative square associated with $F(\term{}f)$.
				
				Conversely, assume that $\un{}_X$ is a section of $\jcomp{}\term{}X$. Then we define a functor $F:\gtcat{\ct{E}}\to\tcat{\ct{E}}$ by $F(A,t)=(A,\comp{}t\arcmp\un{}_X)$. By applying $\jcomp{}$ to $t$ one obtains that $\jcomp{}A\arcmp\comp{}t=\id{X}\arcmp\jcomp{}\term{}X$, hence $\jcomp{}A\arcmp\comp{}t\arcmp\un{}_X=\id{X}\arcmp\jcomp{}\term{}X\arcmp\un{}_X=\id{X}$ and the functor is well defined on objects. Its action on arrows is trivial: it can only map $f:(A,t)\to(B,s)$ to $f:(A,\comp{}t\arcmp\un{}_X)\to(B,\comp{}s\arcmp\un{}_Y)$. The only thing to check is that $\comp{}f\arcmp\comp{}t\arcmp\un{}_X=\comp{}s\arcmp\un{}_Y\arcmp g$. But by hypothesis $f\arcmp t=s\arcmp\term{}g$, so $\comp{}f\arcmp\comp{}t\arcmp\un{}_X=\comp{}s\arcmp\comp{}\term{}g\arcmp\un{}_X=\comp{}s\arcmp\un{}_Y\arcmp g$ by functoriality of $\comp{}$ and naturality of $\un{}$.
			\end{proof}

			\begin{lem}\label{lem:ess-im}
				Let $p:\ct{E}\to\ct{B}$ be a comprehension category with fibred terminal objects $\term{}$. Then $p$ is Lawvere-Ehrhard if and only if the following conditions hold:
				\begin{enumerate}
					\item Given an object $X$ in $\ct{B}$, there is a section $s_X:X\to\comp{}\term{}X$ of the comprehension $\jcomp{}\term{}X$ natural in $X$;
					\item Given $A$ over $X$ and a section $t:X\to\comp{}A$ of the comprehension $\jcomp{}A$, there exist a unique vertical arrow $t^{\#}:\term{}X\to A$ such that $\comp{}t^{\#}\arcmp \un{}_X= t$ and natural in $(A,t)$ in the following sense: whenever $f:A\to B$, $t:X\to\comp{}A$ and $s:Y\to\comp{}B$ are such that $s\arcmp pf=\comp{}f\arcmp t$, then $s^\#\arcmp\term{}pf=f\arcmp t^\#$.
					\[\begin{tikzcd}
						& A \\
						& {\term{}X} \\
						X & {\comp{}A} & X \\
						& {\comp{}\term{}X}
						\arrow["{t^{\#}}", dashed, from=2-2, to=1-2]
						\arrow["t", from=3-1, to=3-2]
						\arrow["{s_X}"', bend right, from=3-1, to=4-2]
						\arrow["{\jcomp{}A}", from=3-2, to=3-3]
						\arrow["{\comp{}t^{\#}}", from=4-2, to=3-2]
						\arrow["{\jcomp{}\term{}X}"', bend right, from=4-2, to=3-3]
					\end{tikzcd}\]
				\end{enumerate}
			\end{lem}
			
			\begin{proof}
				One implication follows by \zcref{prop:comp-term-inv} and \zcref{lem:prf-trm}, except for the naturality. To prove the latter, consider $f:A\to B$, $t:X\to\comp{}A$ and $s:Y\to\comp{}B$ such that $s\arcmp pf=\comp{}f\arcmp t$. Then applying $\term{}$ yields $\term{}s\arcmp\term{}pf=\term{}\comp{}f\arcmp\term{}t$. This, together with the definition of $t^\#$ and $s^\#$ and the naturality of $\coun{}$, proves that the diagram below commutes, proving the claim.
				\[\begin{tikzcd}
					& {\term{}X} &&& {\term{}Y} & \\
					&& A &&& B \\
					{\term{}\comp{}A} &&& {\term{}\comp{}B}
					\arrow["{\term{}pf}", from=1-2, to=1-5]
					\arrow["{t^\#}"{description}, from=1-2, to=2-3]
					\arrow["{\term{}t}"', from=1-2, to=3-1]
					\arrow["{s^\#}", from=1-5, to=2-6]
					\arrow["{\term{}s}"', from=1-5, to=3-4, pos=0.4]
					\arrow["f", from=2-3, to=2-6, crossing over, pos=0.35]
					\arrow["{\coun{}_A}"', from=3-1, to=2-3]
					\arrow["{\term{}\comp{}f}"', from=3-1, to=3-4]
					\arrow["{\coun{}_B}"', from=3-4, to=2-6]
				\end{tikzcd}\]
				
				For the converse, suppose that the conditions hold. We want to show that there is an adjunction $\term{}\dashv\comp{}$. We start by defining the natural transformation $\un{}:\Id{\ct{B}}\to \comp{}\term{}$ whose components are the sections $\un{}_X\colon= s_X$.	Now we can define the counit $\coun{}:\term{}\comp{}\to\Id{\ct{E}}$. First, fix a cleavage of $p$ and consider the generic element of type $A$, $\gel{A}$. It is by definition a section of the comprehension $\jcomp{}\weak{A}{A}$, so by hypothesis we get a unique vertical arrow $\gel{A}^{\#}:\term{}\comp{}A\to\weak{A}{A}$ such that $\comp{}\gel{A}^{\#}\arcmp s_X= \gel{A}$. Finally, we define $\coun{}_A\colon = \weakar{A}{A}\arcmp \gel{A}^{\#}$. 
				\[\begin{tikzcd}
					& {\term{}\comp{}A} \\
					& {\weak{A}{A}} & A \\
					{\comp{}A} \\
					& {\comp{}\weak{A}{A}} & {\comp{}A} \\
					& {\comp{}A} & X
					\arrow["{\gel{A}^{\#}}"', dashed, from=1-2, to=2-2]
					\arrow["{\coun{}_A}", from=1-2, to=2-3]
					\arrow["{\weakar{A}{A}}"', from=2-2, to=2-3]
					\arrow["{{\gel{A}}}", dashed, from=3-1, to=4-2]
					\arrow["{{\id{\comp{}A}}}", bend left, from=3-1, to=4-3]
					\arrow["{{\id{\comp{}A}}}"', bend right, from=3-1, to=5-2]
					\arrow["{\jcomp{}\weak{A}{A}}", from=4-2, to=4-3]
					\arrow["{{\comp{}\weakar{A}{A}}}", from=4-2, to=5-2]
					\arrow["{{\jcomp{}A}}", from=4-3, to=5-3]
					\arrow["{{\jcomp{}A}}"', from=5-2, to=5-3]
				\end{tikzcd}\]
				This definition does not depend on the particular choice of cleavage: there is a unique vertical iso between two different choices of a cleavage, and its mediation with the different reindexing functors does not change the composition. Triangular identities are easy to show. For $X$ in $\ct{B}$, we have that $\coun{}_{\term{}X}=\term{}\jcomp{}\term{}X$. Then one has $\coun{}_{\term{}X}\arcmp\term{}\un{}_X=\id{\term{}X}$ since $\un{}_X$ is a section of $\jcomp{}\term{}X$.	Instead, for $A$ in $\ct{E}$, we have that $\comp{}\coun{}_A\arcmp\un{}_{\comp{}A}=\id{\comp{}A}$ by definition of $\coun{}_A$.
				
				The naturality of $\coun{}$ is easy to prove using naturality in $(A,t)$: given $f:A\to B$, it is sufficient to check that $\reind{f}\arcmp\gel{A}^{\#}=\gel{B}^{\#}\arcmp\term{}\comp{}f$. This is true since $\comp{}\reind{f}\arcmp\gel{A}=\gel{B}\arcmp\comp{}f$ holds.
				
			\end{proof}
			
			\begin{rem}
				The hypothesis of the previous lemma can be loosened a bit: one can drop the requirement of naturalities, asking for only a family of sections and the universal property below. The idea to prove it consists in using the characterization of adjunctions via universal morphisms, which allows to create a left adjoint to $\comp{}$ whose action on objects coincide with $\term{}$. Then it is an easy check that this is a section of $p$ and hence the same as $\term{}$ thanks to its universal property.
			\end{rem}

We can finally prove our characterization of Lawvere-Ehrhard comprehension categories. 

			\begin{thm}\label{thm:ess-im}
				Let $p:\ct{E}\to\ct{B}$ be a comprehension category with fibred terminal objects. 
Then $p$ is Lawvere-Ehrhard if and only if there is an isomorphism $\ftransp{}:\gtsigma\to\tsigma$ in $\Fib(\ct{B})/p$.
			\end{thm}
			
			\begin{proof}
				One implication follows directly from \zcref{prop:prf-trm}. For the converse, it is enough to show that the two conditions of \zcref{lem:ess-im} hold. The first of the two is equivalent to the existence of a morphism $\ttransp:\gtsigma\to\tsigma$ in $\Fib(\ct{B})/p$ by \zcref{lem:prf-trm}. Now it suffices to prove that the existence of an inverse to $\ttransp$ implies the second condition (in fact, it is equivalent).
				
				Consider $A$ over $X$ and $t:X\to\comp{}A$ a section of $\jcomp{}A$. Then we can define $t^\#$ as $\ttransp^{-1}t$.
				Applying $\ttransp$ to the morphism $t^\#:({\term{}X},\id{\term{}X})\to(A,t^\#)$ in $\gtcat{\ct{E}}$ shows that $\comp{}t^\#\arcmp\un{}_X=\ttransp t^\#=t$. The naturality of $t^\#$ corresponds exactly to the fact that $\ttransp^{-1}$ maps morphisms in morphisms.
				
			\end{proof}
			
			This result implies that a large class of comprehension categories are actually Lawvere-Ehrhard.
			
			\begin{cor}\label{cor:full}
Let $p:\ct{E}\to\ct{B}$ be a full comprehension category with fibred terminal objects preserved by the comprehension functor $\jcomp{}$. 
Then, $p$ is a Lawvere-Ehrhard comprehension category. 
			\end{cor}
			
			\begin{proof}
We will denote by $\term{}$ the fibred terminal object functor of $p$, and use the characterization given in \zcref{lem:ess-im}. First, given $X$ in $\ct{B}$ the comprehension $\jcomp{}\term{}X$ has a section: since $\jcomp{}$ preserves fibred terminal objects, $\jcomp{}\term{}X$ is fibred terminal with respect to $\cod$. Then there exist a unique vertical morphism $\id{X}\to \jcomp{}X$, that corresponds exactly to a section $s_X$ of $\jcomp{}\term{}X$.
				\[\begin{tikzcd}
					X & X \\
					{\comp{}\term{}X} & X
					\arrow["{\id{}X}", from=1-1, to=1-2]
					\arrow["{s_X}"', dashed, from=1-1, to=2-1]
					\arrow["{\id{X}}", dashed, from=1-2, to=2-2]
					\arrow["{\jcomp{}\term{}X}"', from=2-1, to=2-2]
				\end{tikzcd}\]
				Secondly, given $A$ over $X$ and a section $t:X\to\comp{}A$ of the comprehension $\jcomp{}A$, we know that $(t\arcmp\jcomp{}\term{}X)\arcmp s_X=t$. Furthermore, we have that $\jcomp{}A\arcmp (t\arcmp\jcomp{}\term{}X)=\jcomp{}\term{}X$, so $(t\arcmp\jcomp{}\term{}X, \id{X}):\jcomp{}\term{}X\to \jcomp{}A$ is a morphism in $\Arr{\ct{B}}$. Then there exist a unique $f:\term{}X\to A$ such that $\jcomp{}f=(t\arcmp\jcomp{}\term{}X, \id{X})$ since $\jcomp{}$ is full and faithful.
				We conclude by \zcref{lem:ess-im}.
			\end{proof}

			\begin{exa}
				Consider the comprehension category given by a family of arrows closed under pullbacks in $\ct{B}$ (see \zcref{ex:cod-gen}). If we moreover suppose that the family contains the identities, then it is a LECC.
				Indeed, this comprehension category satisfies the conditions of \zcref{cor:full}.
			\end{exa}

			We use these results to show that the fibration of (small) discrete fibrations over (small) categories is a LECC, while the fibration of (small) fibration over (small) categories is a JCC which is not a LECC.
			
			\begin{exa}
Let $\Cat$ be the category of small categories and functors, 
$\Fibone$ the category of small fibrations and fibration morphisms, and 
$\Fibdisc$ the full subcategory of $\Fibone$ on discrete fibrations. 
Both $\Fibone$ and $\Fibdisc$ are subcategories of $\Arr{\Cat}$ closed under pullbacks and containing identities, hence 
the restriction of the codomain functor to them is a fibration with fibred terminal objects. 
Moreover, $\Fibdisc$ is a full subcategory of $\Arr{\Cat}$, because in a discrete fibration all arrows of the total category are cartesian. 
Therefore,  by \zcref{cor:full}, $\cod : \Fibdisc\to\Cat$ is a Lawvere-Ehrhard comprehension category. 

Instead, $\cod : \Fibone \to \Cat$ 
is not a particular case of \zcref{ex:cod-gen}, since 
$\Fibone$ is not a full subcategory of $\Arr{\Cat}$. 
In fact, this is an example of a comprehension category which is not Lawvere-Ehrhard. 
In particular, given an object $p:\ct{E}\to\ct{B}$ in $\Fibone$, 
the transpose of $\id{\ct{E}}\in\Hom{\Cat}(\ct{E},\comp{}p)$ should be a morphism $\id{\ct{E}}\to p$ in $\Hom{\Fibone}(\id{\ct{E}},p)$, 
but $(\id{\ct{E}},p)\in\Hom{\Arr{\Cat}}(\id{\ct{E}},p)$ is not a morphism of fibrations since $\id{\ct{E}}$ does not preserve cartesian arrows.
	
		\end{exa}
			
			Both the conditions of \zcref{lem:ess-im} are not necessarily verified. \zcref{ex:cond1} provides a comprehension category in which the first condition of the characterization does not hold. Next example shows that also the second condition is not necessarily verified.
			
			\begin{exa}[{\cite[Exs.~10.4.8]{jacobs1999categorical}}]\label{ex:famsetp}
				Consider the category $\ctsetp$ of pointed sets and the family fibration $\Famfib{\ctsetp}:\Fam{\ctsetp}\to\ctset$, obtained as described in \cite[Definition 1.2.1]{jacobs1999categorical}. This fibration, together with the functor $\jcomp{}:\Fam{\ctsetp}\to\Arr{\ctset}$ that maps $(I,\{X_i\}_{i\in I})$ to $\bigsqcup_{i\in I}X_i\to I$, is a comprehension category (see \cite{jacobs1999categorical}). It also has a fibred terminal object functor $\term{}$, since $\ctsetp$ has a terminal object $(\{*\}, *)$. We can see that this fibration is not Lawvere-Ehrhard: although the first condition of the characterization holds, the second is not satisfied. Indeed, given a set $I$, one has that $\jcomp{}\term{}X$ is an isomorphism, so it has a section. But given an object $(I,\{X_i\}_{i\in I})$ in $\Fam{\ctsetp}$ over $I$ there is a unique vertical arrow $\term{}X\to (I,\{X_i\}_{i\in I})$, while in general there are different sections of $\bigsqcup_{i\in I}X_i\to I$. 
			\end{exa}
			
			We recall that since Lawvere-Ehrhard comprehension is defined by adjointness, there is at most one (up to iso) structure of LECC over a fixed fibration. 
			
			\begin{exa}\label{prop:obzero}
				Consider a fibration $p:\ct{E}\to\ct{B}$ with fibred zero-object $\obzero{p}$. It is a LECC whose comprehension functor is $p$ itself: the adjunctions $p\dashv \obzero{p}$ and $\obzero{p}\dashv p$ hold by (fibred) terminality and initiality of $\obzero{p}$, respectively.
			\end{exa}
			
			We conclude by showing that it is possible to have more than one non-equivalent Jacobs comprehension structures over the same fibration. We do that by endowing the family fibration of \zcref{ex:famsetp} with a Lawvere-Ehrhard comprehension structure. This is possible since the JCC of that example was not in the essential image of $\LtoJ$.
			
			\begin{exa}\label{ex:ourfamsetp}
				Consider the family fibration $\Famfib{\ctsetp}:\Fam{\ctsetp}\to\ctset$. By \zcref{prop:obzero} we have that the triple $(\Famfib{\ctsetp},\term{},\Famfib{\ctsetp})$ is a LECC.
			\end{exa}
			
			\subsection{Faithful comprehension structures}
			
In this section, we compare Jacobs and Lawvere-Ehrhard comprehension categories in the special case where the underlying fibrations is faithful. 
These fibrations essentially correspond to families of preorders instead of families of arbitrary categories. 
They admit a substantially simpler technical treatment but, at the same time, they are general enough for dealing with many applications. 
For example, in type refinement systems \cite{mellies2015functors}, where the subtyping relation is a preorder, or in categorical logic \cite{lawvere1969adjointness, MR0257175} where fibres model the logical entailment relation. 

More in detail, we will show that for Lawvere-Ehrhard comprehension categories faithfulness is equivalent to proof-irrelevance, while this is not the case for arbitrary comprehension categories, where proof-irrelevance is a sufficient but not necessary condition. 
Let us start with a formal definition of proof-irrelevant comprehension category.

\begin{defi} \label{def:proof-irrelevant}
Let $p : \ct{E}\to\ct{B}$ be a comprehension category. 
We say that $p$ is \define{proof-irrelevant} if 
the comprehension functor $\jcomp{}$ is faithful and 
it factors through the subcategory of $\Arr{\ct{B}}$ spanned by monomorphisms. 
\end{defi}

The first condition states that if two arrows $f$ and $g$ in $\ct{E}$ lie over the same arrow of $\ct{B}$, i.e., $pf = pg$, and 
they have the same comprehension, i.e., $\comp{} f = \comp{} g$, then they must be equal. 
Note that this is essentially a restricted form of faithfulness, in fact, 
every comprehension category which is also a faithful fibration satisfies this condition. 
The second condition requires that, for every object $A$ of $\ct{E}$ over $X$ in $\ct{B}$, its comprehension $\jcomp{}A : \comp{}A \to X$ is monic. Keeping in mind that monos have at most one section, this means that every type is inhabited by at most one term. 

The next proposition shows that for comprehension categories proof-irrelevance implies faithfulness. 

			\begin{prop}\label{prop:faith}
Let $p : \ct{E} \to \ct{B}$  be a comprehension category. 
If $p$ is proof-irrelevant then $p$ is faithful. 
			\end{prop}
			
			\begin{proof}
				Consider $f,g:A\to B$ in $\ct{E}$ over $h:X\to Y$. Then $\jcomp{}B\arcmp \comp{}f=h\arcmp \jcomp{}A=\jcomp{}B\arcmp\comp{}g$. Since $\jcomp{}$ is mono, we conclude that $\comp{}f=\comp{}g$. This implies that $\jcomp{}f=\jcomp{}g$, which by faithfulness proves $f=g$.
			\end{proof}
			
			\begin{exa}
				Consider the discrete version of the syntactic fibration of \zcref{ex:typeth2}. This is trivially a faithful JCC, but it is not proof irrelevant. In fact the comprehension morphism $\Gamma, x:A\to \Gamma$ is not monic if there are two different terms of type $A$ in context $\Gamma$.
			\end{exa}

The next theorem shows that for Lawvere-Ehrhard comprehension categories proof-irrelevance is actually equivalent to faithfulness. 
Intuitively, this is due to the fact that in this case there is a tighter connection between  morphisms in the base and in the total category that allows us to prove that in a faithful LECC comprehension arrows are monic. 
This shows how the type theory modelled by a LECC becomes proof irrelevant whenever type morphisms form a preorder between types in the same context.

			\begin{thm}\label{thm:faith-char}
Let $p : \ct{E}\to\ct{B}$ be a Lawvere-Ehrhard comprehension category. Then, $p$ is faithful if and only if $p$ is proof-irrelevant. 
			\end{thm}
			
			\begin{proof}
				One implication consists of \zcref{prop:faith}. For the converse, consider $f,g:X\to \comp{}A$ in $\ct{B}$ such that $p\coun{}_A\arcmp f=p\coun{}_A\arcmp g$, and the transposes $f^{\#},g^{\#}:\term{}X\to A$. They are equal to the composition, respectively, $\coun{}_A\arcmp \term{}f$ and $\coun{}_A\arcmp\term{}g$ by the definition of transposes through counit. These arrows are over $p\coun{}_A\arcmp f=p\coun{}_A\arcmp g$, and by faithfulness they have to be equal. This proves that $p\coun{}_A$ is a mono. Faithfulness trivially implies the other condition required for proof-irrelevance.
			\end{proof}

We conclude this section with two examples
showing that both the conditions of proof-irrelevance are necessary for characterizing faithfulness in Lawvere-Ehrhard comprehension categories. 

\begin{exa}
Consider the fibration $p$ of fibred monoids over sets, i.e. internal monoids in the slices of $\ctset$, which have finite products given by pullbacks.
This fibration has zero-objects, so by \zcref{prop:obzero} it is a LECC. The corresponding Jacobs comprehension functor maps an object $A$ to $\id{pA}$, which is trivially monic. Still, the other condition required for proof-irrelevance is not satisfied: the comprehension functor is not faithful, since $p$ is not and $\jcomp{}$ picks identities.
\end{exa}

\begin{exa} 
Consider the codomain fibration for a category $\ct{C}$ with pullbacks. It is a LECC with $\dom$ as comprehension functor, and the corresponding Jacobs comprehension functor is $\Id{\Arr{\ct{C}}}$. In this case comprehensions are not necessarily monic, but the comprehension functor $\jcomp{}$ is faithful since it is the identity.
\end{exa}

			\section{The free comprehension category over a fibration}
			\label{sect:free-jcomp}

In this section, we will construct the free (Jacobs) comprehension category over a fibration. In type-theoretic terms, this corresponds to adding the rule for context extension (together with the associated rules for  weakening and variables) to a type theory with just the substitution rule. 
We will describe the construction first at the level of objects, 
then we will briefly extend it to the 2-categorical level, finally proving that the obtained 2-functor is left biadjoint to the corresponding forgetful functor. 
We will prove the biadjunction defining its unit and counit and proving triangular identities.
To improve readability, auxiliary technical lemmas are moved to \zcref{app:A}. 

			\begin{notation}
We denote by $\num{n}$ the set $\{ i \in \mathbb{N} \mid 0 < i \leq n \}$, hence, $\num{0} = \emptyset$, $\num{1} = \{1\}$, $\num{2} = \{1,2\}$ and so on.
				Consider a function $g:\num{m}\to \num{n}$. We denote by $g+1$ the function $g+1:\num{m+1}\to \num{n+1}$ such that $(g+1)\rstn_{\num{m}}=g$ and $(g+1) (m+1)= n+1$.
			\end{notation}

We need an auxiliary construction in order to define the free comprehension category over a fibration. 
In particular, we will consider a 2-functor $\Fp:\Fib\to\Fib$ that freely adds finite products to the fibres of a fibrations. 
More precisely, this sends a fibration $p : \ct{E}\to\ct{B}$ to the fibration $\Fp(p) : \FFFP{p} \to \ct{B}$ defined as follows. 
The objects of $\FFFP{p}$ are pairs $(X,\vec{A})$ with $X$ in $\ct{B}$ and $\vec{A}$ a finite list of objects in $\fibre{\ct{E}}{X}$ and 
morphisms $f:(X,\vec{A})\to(Y,\vec{B})$ are triples $(\fp{f},\sp{f},\tp{f})$ 
where 
\begin{itemize} 
	\item $\fp{f}:X\rightarrow Y$ is an arrow in $\ct{B}$;
	\item $\sp{f}:\num{m}\rightarrow \num{n}$ is a function, where $m$ is the length of $\vec{B}$ and $n$ is the length of $\vec{A}$; 
	\item $\tp{f}= \{f_i\}_{i\in \num{m}}$ is a family of arrows in $\ct{E}$ such that $f_i:A_{\sp{f}(i)}\rightarrow B_i$.
\end{itemize} 
Finally, $\Fp(p)$ maps $f:(X,\vec{A})\to(Y,\vec{B})$ to $\fp{f}:X\to Y$.
A cartesian lifting of $f:X\to Y$ at $(Y,\vec{A})$ is obtained by taking the family of cartesian liftings with respect to $p$, that is to say it is $g:(X,\reind{f}\vec{A})\to (Y,\vec{A})$ such that $\fp{g}=f$, $\sp{g}=\id{\num{n}}$ and $\tp{g}=\{\carar{f}{A_i}\}_{i=1..n}$.
Given a morphism of fibrations $F:p\to q$, we define $\Fp(F):\Fp(p)\to\Fp(q)$ in the obvious way: it maps $(X,\vec{A})$ to $(\fb{F}X,\ft{F}\vec{A})$. Given a 2-cell $\alpha:F\to G$ in $\Fib$, we define $\Fp(\alpha)$ analogously.

Let $p : \ct{E}\to\ct{B}$ be  a fibration  and 
consider the following 2-pullback in $\Cattwo$: 
			\[\begin{tikzcd}
				{\JDomE{p}} & {\ct{E}} \\
				{\FFFP{p}} & {\ct{B}}
				\arrow[from=1-1, to=1-2]
				\arrow["{\JFib{p}}"', from=1-1, to=2-1]
				\arrow["p", from=1-2, to=2-2]
				\arrow["{{\fffp{p}}}"', from=2-1, to=2-2]
			\end{tikzcd}\]

Since the pullback of a fibration along any functor is again a fibration \cite{StreicherT:fibc}, 
we know that the functor $\JFib{p}$ is a fibration. 
More explicitly, objects in $\JDomE{p}$ are pairs $((X,\vec{A}), A_{n+1})$ with the first element in $\FFFP{p}$ and the second in $\ct{E}$ over $X$.
\footnote{To avoid ambiguity we make abundant use of parenthesis, sacrificing readability: $(X,(\vec{A}, A_{n+1}))$ indicates an object of the base category, while $((X,\vec{A}),A_{n+1})$ indicates an object in the fibre over $(X,\vec{A})$.}
The morphisms in this category are pairs of arrows $(f,g)$ such that $pg=\fp{f}$. Furthermore, $\JFib{p}$ is the first projection.

Intuitively, an object $(X,\vec{A})$ in $\FFFP{p}$ represents a context obtained by extending an original context $X$ in $\ct{B}$ with a list of types in that context. 
Similarly, an object $((X,\vec{A}),A_{n+1})$ of $\JDomE{p}$, i.e., a type in the context $(X,\vec{A})$, is just the weakening  of an original type $A$ over $X$ with  the new variables in $\vec{A}$. 
Following this intuitive reading of $\JFib{p}$, 
it is easy to see that it naturally support a context extension (a.k.a. comprehension) operation: 
the comprehension of a type $((X,\vec{A}),A_{n+1})$ is given by adding the specified type $A_{n+1}$ to the list $\vec{A}$ obtaining the context $(X, (\vec{A},A_{n+1}))$.
In this way, the objects of $\FFFP{p}$ can also be described as the result of the iterative application of comprehension and weakening: 
to construct $(X, A_1,\ldots, A_n)$, we start from $X$ and take the comprehension of $A_1$ obtaining $(X,A_1)$, 
then we weaken $A_2$, bringing it over $(X,A_1)$, and take its comprehension yielding $(X, (A_1 , A_2))$ and so on.

Formally, we can give to $\JFib{p}$ the structure of a comprehension category by defining 
a functor $\JFC{p}:\JDomE{p}\to\Arr{\FFFP{p}}$ as follows:

\begin{itemize}
	\item it maps objects $Y=((X,\vec{A}),A_{n+1})$ into $\JFC{p} Y:(X,(\vec{A}, A_{n+1}))\rightarrow(X,\vec{A})$ given by 
	\[
	\fp{\JFC{p} Y}=\id{X},
	\hspace{7ex}
	\sp{\JFC{p} Y}=\ini{\num{n}}{\num{n+1}},
	\hspace{7ex}
	\tp{\JFC{p} Y}=\{\id{A_i}\}_{i\in \num{n}}.
	\]
	
	\item it maps arrows $f:Y\rightarrow Z$, where $f=(g,h)$ and $Z=((X',\vec{B}), B_{m+1})$, into the square 
	\[\begin{tikzcd}
		{(X,(\vec{A}, A_{n+1}))} & {(X',(\vec{B}, B_{m+1}))} \\
		{(X,\vec{A})} & {(X',\vec{B})}
		\arrow["\carr{f}", from=1-1, to=1-2]
		\arrow["\JFC{p} Y"', from=1-1, to=2-1]
		\arrow["\JFC{p} Z", from=1-2, to=2-2]
		\arrow["g"', from=2-1, to=2-2]
	\end{tikzcd}\]
	where $\carr{f}\colon=(\fp{g},\sp{g}+1,\tp{g}\sqcup h)$. Sometimes we will use the notation $(g,\carr{f})$ for $\JFC{p}f$, since the vertical sides of the square are clear by the context.
	
\end{itemize}

			\begin{prop}
				Let $p:\ct{E}\to\ct{B}$ be a fibration. Then $\JFib{p}$ together with $\JFC{p}:\JDomE{p}\to \Arr{\FFFP{p}}$ is a comprehension category. 
			\end{prop}
			
			\begin{proof}
				It is easy to see that the following square commutes.
				
				\[\begin{tikzcd}
					{(X,(\vec{A}, A_{n+1}))} & {(X',(\vec{B}, B_{m+1}))} \\
					{(X,\vec{A})} & {(X',\vec{B})}
					\arrow["\carr{f}", from=1-1, to=1-2]
					\arrow["\JFC{p} Y"', from=1-1, to=2-1]
					\arrow["\JFC{p} Z", from=1-2, to=2-2]
					\arrow["g"', from=2-1, to=2-2]
				\end{tikzcd}\]
				
				Furthermore given a pair of composable arrows $f,f'$ one has that $\carr{f}\arcmp \carr{f'}=\carr{(f\arcmp f')}$ and $\carr{\id{Y}}=\id{\dom(\JFC{p} Y)}$. This implies that $\JFC{p}:\JDomE{p}\rightarrow \Arr{\FFFP{p}}$ is a functor. 
				
				Now we only need to prove that $\jcomp{\JFib{p}}$ preserves cartesian arrows. Given $f=(f_1,f_2)$ cartesian in $\JDomE{p}$, we have that $f_2$ is cartesian as well. Consider two morphisms $g:(X'',\vec{C})\rightarrow (X,\vec{A})$ and $h:(X'',\vec{C})\rightarrow (X', (\vec{B}, B_{m+1}))$ such that $f_2\arcmp g=\JFC{p} Z\arcmp h$.
				Then one can define $u:(X'', \vec{C})\rightarrow (X, (\vec{A}, A_{n+1}))$ as follows:
				
				\begin{align*}
					\fp{u} = \fp{g} &&
					\sp{u} : \num{n+1} \rightarrow \num{k} &&
					\tp{u} = \tp{g} \sqcup (\tp{h})_{m+1}
				\end{align*}
				
				where $\sp{u}(i)=\sp{g}(i)$ for $i\in \num{n}$ and $\sp{u}(n+1)= \sp{h}(m+1)$.
				Clearly $u$ is the unique arrow that makes the diagram below to commute, making $\JFC{p} f$ a pullback. 
				
				\[\begin{tikzcd}
					{(X'', \vec{C})} \\
					& {(X, (\vec{A}, A_{n+1}))} & {(X', (\vec{B}, B_{m+1}))} \\
					& {(X,\vec{A})} & {(X', \vec{B})}
					\arrow["u", dashed, from=1-1, to=2-2]
					\arrow["h", bend left, from=1-1, to=2-3]
					\arrow["g"', bend right, from=1-1, to=3-2]
					\arrow["{\carr{f}}", from=2-2, to=2-3]
					\arrow["{\JFC{p} Y}"', from=2-2, to=3-2]
					\arrow["{\JFC{p} Z}", from=2-3, to=3-3]
					\arrow["{f_2}"', from=3-2, to=3-3]
				\end{tikzcd}\] 
			\end{proof}
			
			We are now going to show that $\JFib{p}$ is the free comprehension category over an arbitrary fibration $p$. In particular, we will extend this construction to a 2-functor $\JFree:\Fib\to\JComp$ and prove that it is left bi-adjoint to the forgetful 2-functor $\forgJ:\JComp\to\Fib$.
			
			Given $F:p\to p'$ a morphism in $\Fib$, we define $\fb{\JFun{F}}\colon =\ft{\Fffp{F}}$ and $\ft{\JFun{F}}$ as the unique arrow given by the universal property of the 2-pullback defining $\JDom{\ct{E}'}{q}$, since $\fffp{p'}\arcmp\fb{\JFun{F}}\arcmp\JFib{p}=p'\arcmp \ft{F}\arcmp\pi_2$.
			\[\begin{tikzcd}
				{\ct{E}} \\
				{\JDomE{p}} & {\JDom{\ct{E}'}{q}} & {\ct{E}'} \\
				{\FFFP{p}} & {\FFFP{q}} & {\ct{B'}} \\
				{\ct{B}}
				\arrow["\ft{F}", from=1-1, to=2-3]
				\arrow["\pi_2", from=2-1, to=1-1]
				\arrow["{\ft{\JFun{F}}}", dashed, from=2-1, to=2-2]
				\arrow["{\JFib{p}}"', from=2-1, to=3-1]
				\arrow[from=2-1, to=3-2]
				\arrow[from=2-2, to=2-3]
				\arrow["{\JFib{q}}", from=2-2, to=3-2]
				\arrow["q", from=2-3, to=3-3]
				\arrow["{\fb{\JFun{F}}}"', from=3-1, to=3-2]
				\arrow["{\fffp{p}}"', from=3-1, to=4-1]
				\arrow["{\fffp{q}}", from=3-2, to=3-3]
				\arrow["\fb{F}"', from=4-1, to=3-3]
			\end{tikzcd}\]
			Now consider a 2-cell $\alpha:F\to G$ in $\Fib$. We define $\fb{\JFun{\alpha}}\colon = \ft{\Fffp{\alpha}}$ and $\ft{\JFun{\alpha}}$ as the unique 2-cell given by the universal property of the 2-pullback. 

			Explicitly, given $f:(X,\vec{A})\rightarrow (X',\vec{B})$ in $\FFFP{p}$ one has $\fb{\JFun{F}}(f):(\fb{F}X,\ft{F}\vec{A})\rightarrow (\fb{F}X',\ft{F}\vec{B})$ given by:
			
			\begin{align*}
				\fp{\fb{\JFun{F}}(f)} = \fb{F}\fp{f} &&
				\sp{\fb{\JFun{F}}(f)} = \sp{f} &&
				\tp{\fb{\JFun{F}}(f)} = \{\ft{F}(\tp{f})_i\}_{i\in\num{n}}
			\end{align*}		
			Moreover, given $f=(g,h):((X,\vec{A}), A_{n+1})\rightarrow ((X',\vec{B}), B_{m+1})$ in $\JDomE{p}$, one has $\ft{\JFun{F}}(f)=(\fb{\JFun{F}}(g), \ft{F}h):((\fb{F}X,\ft{F}\vec{A}), \ft{F}A_{n+1})\rightarrow ((\fb{F}X',\ft{F}\vec{B}), \ft{F}B_{n+1})$. This means that the functor does not change the length of the list, and applies the functors $\fb{F},\ft{F}$ to the corresponding components.

			\begin{lem}\label{lem:J2fun}
				$\JFree:\Fib\to\JComp$ is a 2-functor.
			\end{lem}
			
			\begin{proof}
				It is easy to see that $\ft{\JFun{F}}$ preserves comprehensions on the nose.
				In fact, consider an arbitrary arrow $f=(g,h):((X,\vec{A}), A_{n+1})\rightarrow ((X',\vec{B}), B_{m+1})$ in $\JDomE{p}$. Then $\JFC{q}\arcmp \ft{\JFun{F}}(f)$ is the following square:
				
				\[\begin{tikzcd}
					{(\fb{F}X,(\ft{F}\vec{A}, \ft{F}A_{n+1}))} & {(\fb{F}X',(\ft{F}\vec{B}, \ft{F}B_{m+1}))} \\
					{(\fb{F}X,\ft{F}\vec{A})} & {(\fb{F}X',\ft{F}\vec{B})}
					\arrow["\carr{\ft{\JFun{F}}(f)}", from=1-1, to=1-2]
					\arrow["\JFC{q} \ft{\JFun{F}}Y", from=1-1, to=2-1]
					\arrow["\JFC{q} \ft{\JFun{F}}Z", from=1-2, to=2-2]
					\arrow["\fb{\JFun{F}}(g)", from=2-1, to=2-2]
				\end{tikzcd}\]
				
				Instead, $\Arr{(\fb{\JFun{F}})}\arcmp \JFC{p} (f)$ is the following:
				
				\[\begin{tikzcd}
					{(\fb{F}X,(\ft{F}\vec{A}, \ft{F}A_{n+1}))} & {(\fb{F}X',(\ft{F}\vec{B}, \ft{F}B_{m+1}))} \\
					{(\fb{F}X,\ft{F}\vec{A})} & {(\fb{F}X',\ft{F}\vec{B})}
					\arrow["\fb{\JFun{F}}(\carr{f})", from=1-1, to=1-2]
					\arrow["\fb{\JFun{F}}\JFC{p} Y", from=1-1, to=2-1]
					\arrow["\fb{\JFun{F}}\JFC{p} Z", from=1-2, to=2-2]
					\arrow["\fb{\JFun{F}}(g)", from=2-1, to=2-2]
				\end{tikzcd}\]
				
				First, $\JFC{q} \ft{\JFun{F}}Y=\fb{\JFun{F}}\JFC{p} Y$ because every component is the same. Furthermore, $\carr{\ft{\JFun{F}}(f)}=\fb{\JFun{F}}(\carr{f})$ again because every component is the same. This proves that $\JFun{F}$ is a morphism of comprehension categories.
				
				Moreover, it is easy to see also that $\JFun{\alpha}:\JFun{F}\to\JFun{G}$ is a 2-cell in $\Fib$ and that the equality $\lwhisk{\JFC{q}}\ft{\JFun{\alpha}}=\Arr{(\fb{\JFun{\alpha}})}\rwhisk{\JFC{p}}$ holds, proving that it is a 2-cell in $\JComp$.
				
				Finally, the preservation of compositions and identities follows by the universal property of the 2-pullback.
			\end{proof}
			
			\begin{exa}\label{ex:cond1}
				Let $p:\ct{E}\to\ct{B}$ be a fibration with fibred terminal objects $\term{}$. Then its completion $\LFib{p}$ is an example of comprehension category with fibred terminal objects which does not satisfy the first condition of \zcref{lem:ess-im}. In fact, given $(X,())$ in $\FFFP{p}$ one has that $\jcomp{\LFib{p}}\term{\LFib{p}}(X,())$ does not have sections, since there are no functions $\num{0}\to\num{1}$.
			\end{exa}

			Now we can show that $\JFree:\Fib\to\JComp$ and $\forgJ:\JComp\to \Fib$ are a bi-adjoint pair. The first step to do this is defining unit and counit of the bi-adjunction. The naturality diagrams will not commute strictly in general, but only up to iso. Hence we need to require the unit and the counit to be pseudo-natural, meaning that we not only need to index them on 0-cells, but also on 1-cells: the latter components will be the invertible 2-cells filling the naturality squares.
			Given a fibration $p:\ct{E}\rightarrow \ct{B}$, we define a fibration morphism $\twont{\eta}_p:p\rightarrow \forgJ\JFib{p}$ using the universal property of the 2-pullback applied to the terminal functor of $\fffp{p}$ and the identity of the total category of $p$, which are respectively $\term{\fffp{p}}$ and $\Id{\ct{E}}$.
			\[\begin{tikzcd}
				{\ct{E}} && \\
				{\ct{B}} & {\JDomE{p}} & {\ct{E}} \\
				& {\FFFP{p}} & {\ct{B}}
				\arrow["p"', from=1-1, to=2-1]
				\arrow["\ft{\twont{\eta}_p}"{description}, dashed, from=1-1, to=2-2]
				\arrow["{\Id{\ct{E}}}", from=1-1, to=2-3, bend left]
				\arrow["{\term{\fffp{p}}}"', from=2-1, to=3-2]
				\arrow[from=2-2, to=2-3]
				\arrow["{\JFib{p}}", from=2-2, to=3-2]
				\arrow["p", from=2-3, to=3-3]
				\arrow["{\fffp{p}}"', from=3-2, to=3-3]
			\end{tikzcd}\]
			Furthermore, given $F:p\to q$ a fibration morphism, the 2-cell $\twont{\eta}_F$ is the identity, since $\fb{\twont{\eta}_q}\arcmp \fb{F}=\fb{\LFun{F}}\arcmp\fb{\twont{\eta}_p}$ and $\ft{\twont{\eta}_q}\arcmp \ft{F}=\ft{\LFun{F}}\arcmp\ft{\twont{\eta}_p}$. 
			
Explicitly, $\twont{\eta}_p$ acts as follows:
$\fb{\twont{\eta}_p}:\ct{B}\rightarrow\FFFP{p}$ is the fibred terminal object functor $\term{\fffp{p}}$, and $\ft{\twont{\eta}_p}:\ct{E}\rightarrow \JDomE{p}$ sends an object $A$ to itself over its basis, i.e. to $((pA,()), A)$. Its action on arrows is defined in the obvious way.
			\[\begin{tikzcd}
				{\ct{E}} & {\JDomE{p}} \\
				{\ct{B}} & {\FFFP{p}}
				\arrow["{\ft{\twont{\eta}_p}}", from=1-1, to=1-2]
				\arrow["p"', from=1-1, to=2-1]
				\arrow["{\JFib{p}}", from=1-2, to=2-2]
				\arrow["\fb{\twont{\eta}_p}"', from=2-1, to=2-2]
			\end{tikzcd}\]
	
The definition of the counit 
$\twont{\epsilon}$ of the bi-adjunction is more delicate: 
we will define the functors involved by induction on the length of the list which they apply to. 
In particular, we will construct a family of morphisms which represent the iterated comprehensions as described before, and use them to define the counit.

			\begin{defi}\label{def:comp-maps}
				Let $p:\ct{E}\rightarrow \ct{B}$ be a comprehension category. 
				First, given an object $(X,\vec{A})$ with length $n$ we define a family of arrows (and their domains) $\C{0}{k}:\fb{\twont{\epsilon}_p}(X,\vec{A}\rstn_k)\rightarrow X$ for any $k\leq n$ by induction on $k$:
				\begin{itemize}[align=left]
					\item[\emph{k=0:}] $\C{0}{0}=\id{X}:X\rightarrow X$;
					\item[\emph{k+1:}] Let $A_{k+1}^*$ be the reindexing of $A_{k+1}$ along $\C{0}{k}$. Then $\C{0}{k+1}=\C{0}{k}\arcmp\chi A_{k+1}^*$.
				\end{itemize}	
				Then one can define $\C{i}{k}:\fb{\twont{\epsilon}_p}(X,\vec{A}\rstn_k)\rightarrow \fb{\twont{\epsilon}_p}(X,\vec{A}\rstn_i)$ for $i\leq k$ by induction on $k-i$:
				\begin{itemize}[align=left]
					\item[\emph{k-i=0:}] $\C{k}{k}=\id{\fb{\twont{\epsilon}_p}(X,\vec{A}\rstn_k)}$;
					\item[\emph{k-i+1:}] $\C{i-1}{k}= \chi A_{i}^*\arcmp \C{i}{k}$.
				\end{itemize}
			\end{defi}
			
Notice that, in the above definition, there is a slight abuse of notation: 
there are two arrows denoted by $\C{0}{k}$. 
However, this is not problematic as 
one can easily see by induction that they actually coincide.
Moreover, one has that arrows $\C{i}{k}$ are well-behaved under composition in the following sense.
			
			\begin{lem}\label{lem:C-coherence}
				For any $i\leq k\leq j$, we have $\C{i}{k}\arcmp \C{k}{j} = \C{i}{j}$. 
			\end{lem}
			
			\begin{proof}
				First, we want to prove that for any $i\leq k$ we have that $\C{i}{k+1}=\C{i}{k}\arcmp \jcomp{}\reind{A_{k+1}}$. This follows by a straightforward induction on $k-i$.
				
				Given that, the result follows easily by induction on $j-k$.

			\end{proof}
			
			Intuitively, the object $\fb{\twont{\epsilon}_p}(X,\vec{A})$ is given by taking the reindexing of $A_1$ along the identity, then the reindexing of $A_2$ along the comprehension of (the reindexing of) $A_1$, and so on until we get to the domain of the comprehension of (the reindexing of) $A_n$. The next lemma will let us define the action of $\fb{\twont{\epsilon}_p}$ on morphisms.
			
			\begin{lem}\label{lem:Cdef}
				Given $f:(X,\vec{A})\rightarrow(X',\vec{B})$, there is a unique family of morphisms $\{g_{i}:\fb{\twont{\epsilon}_p}(X,\vec{A})\to \fb{\twont{\epsilon}_p}(X',\vec{B}\rstn_i)\}_{i=0..m}$, where $n,m$ are respectively the lengths of $\vec{A},\vec{B}$, such that $\C{0}{i}\arcmp g_i=\fp{f}\arcmp \C{0}{n}$ and $g_i=\jcomp{}\reind{B_{i+1}}\arcmp g_{i+1}$.
				\[\begin{tikzcd}
					& {\fb{\twont{\epsilon}_p}(X',(\vec{B},B_{i+1}))} \\
					{\fb{\twont{\epsilon}_p}(X,\vec{A})} & {\fb{\twont{\epsilon}_p}(X',\vec{B})} \\
					\\
					X & {X'}
					\arrow["{\chi B_{i+1}^*}", from=1-2, to=2-2]
					\arrow["{g_{i+1}}", from=2-1, to=1-2, bend left]
					\arrow["{g_i}"', from=2-1, to=2-2]
					\arrow["{\C{0}{n}}"', from=2-1, to=4-1]
					\arrow["{\C{0}{i}}", from=2-2, to=4-2]
					\arrow["{\fp{f}}"', from=4-1, to=4-2]
				\end{tikzcd}\]
			\end{lem}
			
			\begin{proof}
				We define $g_i$ by induction. 
				\begin{itemize}[align=left]
					\item[\emph{i=0:}] $g_0=\fp{f}\arcmp\C{0}{n}:\fb{\twont{\epsilon}_p}(X,\vec{A})\rightarrow X'$;
					\item[\emph{i+1:}] Consider the square $\chi h$ in $\ct{B}$, where $h:B_{i+1}^{*}\to B_{i+1}$ is cartesian over $\C{0}{i}$: this is a pullback. Then consider the diagram
					\[\begin{tikzcd}
						&& {} & {\fb{\twont{\epsilon}_p}(X,\vec{A}\rstn_{\sp{f}(i+1)})} \\
						{\fb{\twont{\epsilon}_p}(X,\vec{A})} && {\comp{} A_{\sp{f}(i+1)}} & {\fb{\twont{\epsilon}_p}(X,\vec{A}\rstn_{\sp{f}(i+1)-1})} \\
						& {\fb{\twont{\epsilon}_p}(X',(\vec{B}\rstn_i,B_{i+1}))} & {\comp{} B_{i+1}} & X \\
						& {\fb{\twont{\epsilon}_p}(X',\vec{B}\rstn_i)} & {X'}
						\arrow[from=1-4, to=2-3]
						\arrow["{\C{\sp{f}(i+1)-1}{\sp{f}(i+1)}}", from=1-4, to=2-4]
						\arrow["{\C{\sp{f}(i+1)}{n}}", bend left, from=2-1, to=1-4]
						\arrow["g_{i+1}"', dashed, from=2-1, to=3-2]
						\arrow["g_i"', bend right, from=2-1, to=4-2]
						\arrow["{\comp{}(\tp{f})_{i+1}}"', from=2-3, to=3-3]
						\arrow[from=2-3, to=3-4]
						\arrow["{\C{0}{\sp{f}(i+1)-1}}", from=2-4, to=3-4]
						\arrow[from=3-2, to=3-3]
						\arrow["\chi B_{i+1}^*", from=3-2, to=4-2]
						\arrow["{\chi B_{i+1}}"', from=3-3, to=4-3]
						\arrow["{\fp{f}}"', from=3-4, to=4-3]
						\arrow["{\C{0}{i}}", from=4-2, to=4-3]
					\end{tikzcd}\]
					where $g_{i+1}$ is the unique arrow given by the universal property of the pullback. 
				\end{itemize}
				By construction we have that $g_i=\jcomp{}\reind{B_{i+1}}\arcmp g_{i+1}$. By inductive hypothesis we know $\C{0}{i}\arcmp g_i=\fp{f}\arcmp \C{0}{n}$, hence $\C{0}{i+1}\arcmp g_{i+1}=\C{0}{i}\arcmp \C{i}{i+1}\arcmp g_{i+1}=\C{0}{i}\arcmp g_i=\fp{f}\arcmp \C{0}{n}$ remembering that $\C{i}{i+1}=\jcomp{}\reind{B_{i+1}}$.
				The unicity is forced for $i=0$ by the required equation $g_0=\fp{f}\arcmp\C{0}{n}$ and for $i>0$ by the universal property of the pullback together with the other required equation.
			\end{proof}
			
We use the family just described to define the action of $\twont{\epsilon}_p$ on morphisms, setting 
\[ \fb{\twont{\epsilon}_p}f\colon = g_m \] 
It is not difficult to see that, 
given two composable morphisms $f$ and $h$, we have that $\fb{\twont{\epsilon}_p}(f\arcmp h)=\fb{\twont{\epsilon}_p}f\arcmp \fb{\twont{\epsilon}_p}h$.
This is a consequence of \zcref{lem:Cdef}: it is enough to notice that the family corresponding to $f$ precomposed with $\fb{\twont{\epsilon}_p}h$ satisfies the required conditions.	Furthermore, if $f$ is the identity, then the family $\{\C{i}{n}\}_{i=0..n}$ satisfies the required conditions, hence $\fb{\twont{\epsilon}_p}(\id{(X,\vec{A})})=\C{n}{n}=\id{\fb{\twont{\epsilon}_p}(X,\vec{A})}$.
Therefore, $\fb{\twont{\epsilon}_p}:\FFFP{p}\rightarrow \ct{B}$ is indeed a functor. 

			Given $Y=((X,\vec{A}), A_{n+1})$ in $\JDomE{p}$, we define $\ft{\twont{\epsilon}_p}(Y)$ as the reindexing of $A_{n+1}$ along $\C{0}{n}$. Given also $Z=((X',\vec{B}), B_{m+1})$ and an arrow $f=(g,h):Y\to X$, we define $\ft{\twont{\epsilon}_p}f$ as the unique arrow given by cartesianity of the cartesian lifting of $\C{0}{m}$ at $B_{m+1}$ over $\fb{\twont{\epsilon}_p}g$.
			\[\begin{tikzcd}
				{\ft{\twont{\epsilon}_p}Y} & {A_{n+1}} \\
				{\ft{\twont{\epsilon}_p}Z} & {B_{m+1}} \\
				{\fb{\twont{\epsilon}_p}(X,\vec{A})} & X \\
				{\fb{\twont{\epsilon}_p}(X',\vec{B})} & {X'}
				\arrow[from=1-1, to=1-2]
				\arrow["{\ft{\twont{\epsilon}_p}f}"', dashed, from=1-1, to=2-1]
				\arrow["h", from=1-2, to=2-2]
				\arrow[from=2-1, to=2-2]
				\arrow["{\C{0}{n}}", from=3-1, to=3-2]
				\arrow["{\fb{\twont{\epsilon}_p}g}"', from=3-1, to=4-1]
				\arrow["{\fp{g}}", from=3-2, to=4-2]
				\arrow["{\C{0}{m}}"', from=4-1, to=4-2]
			\end{tikzcd}\]

			The pair $\twont{\epsilon}_p=(\fb{\twont{\epsilon}_p},\ft{\twont{\epsilon}_p})$ is a morphism of comprehension categories. In particular, comprehension is preserved on the nose.
			
			\begin{prop}
				$\twont{\epsilon}_p$, together with the identity 2-cell, is a morphism in $\JComp$.
			\end{prop}
			
			\begin{proof}
				First we need to show that it is a morphism of fibrations, so that it preserves cartesian morphisms. Consider a cartesian arrow $(f,g):((X,\vec{A}),A_{n+1})\to ((Y,\vec{B}), B_{m+1})$. We have that $g:A_{n+1}\to B_{m+1}$ is cartesian (over $\fp{f}$). Given then $h:\ft{\twont{\epsilon}_p}((X,\vec{A}),A_{n+1})\to A_{n+1}$ and $k:\ft{\twont{\epsilon}_p}((Y,\vec{B}), B_{m+1})\to B_{m+1}$ the cartesian liftings of, respectively, $\C{0}{n}$ and $\C{0}{m}$, we have that $\ft{\twont{\epsilon}_p}f$ is given by cartesianity of $k$. Then $\C{0}{m}\arcmp\ft{\twont{\epsilon}_p}f=\fp{f}\arcmp\C{0}{n}$ is cartesian, and we conclude that $\ft{\twont{\epsilon}_p}f$ is cartesian as well.
				
				Now we only need to show that it preserves comprehensions strictly, i.e. $(\Arr{(\fb{\twont{\epsilon}_p})}\arcmp\JFC{p})=(\jcomp{p}\arcmp \ft{\twont{\epsilon}_p})$. This is a straightforward consequence of \zcref{rmk:lemma}, \zcref{lem:jcoun1} and \zcref{lem:jcoun2}.
			\end{proof}
			
			We also need to define the counit indexed by morphisms of comprehension categories. Consider $F:p\to q$, together with $\alpha:(\Arr{(\fb{F})}\arcmp \jcomp{p})\nt (\jcomp{q}\arcmp \ft{F})$ a morphism in $\JComp$, $((X,\vec{A}), A_{n+1})$ an object in $\JDomE{p}$, and fix a cleavage of $p$ and $q$, respectively. We define simultaneously $(\fb{\Jcoun{F}})_{(X,\vec{A})}$ and $(\ft{\Jcoun{F}})_{((X,\vec{A}), A_{n+1})}$ by induction on the length $n$ of $\vec{A}$. For $1\leq k\leq n$, consider the families of cartesian arrows $i_k:A_k^*\to A_k$ over $\C{0}{k-1}$ and $j_k:(\ft{F}A_{k})^*\to \ft{F}A_k$ cartesian over $d^0_{k-1}$, where we denote by $d^i_j$ the maps defined in \zcref{def:comp-maps} w.r.t. $q$. 
			
			For $n=0$, we set $(\fb{\Jcoun{F}})_{(X,())}\colon=\id{\fb{F}X}$ and $(\ft{\Jcoun{F}})_{((X,()), A_{1})}$ as the unique vertical arrow obtained by cartesianity. Notice that $(\ft{\Jcoun{F}})_{((X,()), A_{1})}$ is over $(\fb{\Jcoun{F}})_{(X,())}$. The latter is trivially invertible, being the identity. The former is invertible since $\ft{F}$ preserves cartesianity.
			\[\begin{tikzcd}
				{\ft{F}A_1^*} \\
				{(\ft{F}A_1)^*} & {\ft{F}A_1} \\
				{\fb{F}X} & {\fb{F}X}
				\arrow["{(\ft{\Jcoun{F}})_{((X,()), A_{n+1})}}"', dashed, from=1-1, to=2-1]
				\arrow["{\ft{F}i_1}", from=1-1, to=2-2]
				\arrow["{j_1}"', from=2-1, to=2-2]
				\arrow["{\id{\fb{F}X}}", from=3-1, to=3-2]
			\end{tikzcd}\]
				
			For $n+1$, we set $(\fb{\Jcoun{F}})_{(X,\vec{A})}\colon=\comp{q}(\ft{\Jcoun{F}})_{((X,\vec{A}\rstn_n), A_{n+1})}\arcmp\psi_{A_{n+1}^*}$, where $\psi:\fb{F}\comp{p}\nt\comp{q}\ft{F}$ is $\lwhisk{\dom}\alpha$. It is iso since it is composition of two invertible morphisms ($(\ft{\Jcoun{F}})_{((X,\vec{A}\rstn_n), A_{n+1})}$ is iso by inductive hypothesis). Then we set $(\ft{\Jcoun{F}})_{((X,\vec{A}), A_{n+2})}$ as the unique arrow over $(\fb{\Jcoun{F}})_{(X,\vec{A})}$ given by cartesianity. Again, this is invertible because $\ft{F}$ preserves cartesian arrows.
			\[\begin{tikzcd}
				& {\ft{F}A_{n+2}^*} \\
				& {(\ft{F}A_{n+2})^*} & {\ft{F}A_{n+2}} \\
				& {\fb{F}\comp{p}A_{n+1}^*} \\
				{\comp{q}\ft{F}A_{n+1}^*} && {\fb{F}X} \\
				& {\comp{q}(\ft{F}A_{n+1})^*}
				\arrow["{{(\ft{\Jcoun{F}})_{((X,\vec{A}), A_{n+2})}}}"', dashed, from=1-2, to=2-2]
				\arrow["{{\ft{F}i_{n+2}}}", from=1-2, to=2-3]
				\arrow["{{j_{n+2}}}"', from=2-2, to=2-3]
				\arrow["\psi_{A_{n+1}^*}"', from=3-2, to=4-1]
				\arrow["{{\fb{F}\C{0}{n+1}}}", from=3-2, to=4-3]
				\arrow["{{(\fb{\Jcoun{F}})_{(X,\vec{A})}}}"', from=3-2, to=5-2]
				\arrow["{{\comp{q}(\ft{\Jcoun{F}})_{((X,\vec{A}\rstn_n), A_{n+1})}}}"', from=4-1, to=5-2]
				\arrow["{{d^0_{n+1}}}"', from=5-2, to=4-3]
			\end{tikzcd}\]
			
			Let us show that the downside diagram commutes, which allows us to use cartesianity. By \zcref{lem:C-coherence} and using the definition of $\C{0}{n+1}$, we know that $\C{0}{n}\arcmp \jcomp{p}A_{n+1}^*=\C{0}{n+1}$, and analogously $d^0_n\arcmp \jcomp{q}(\ft{F}A_{n+1})^*=d^0_{n+1}$.
			Furthermore, we have $\jcomp{q}\ft{F}A_{n+1}^* \arcmp \psi_{A_{n+1}^*} = \fb{F}\jcomp{p}A_{n+1}^*$ since it is a component of $\alpha$, and $\jcomp{q}(\ft{F}A_{n+1}^*) \arcmp \comp{q}(\ft{\Jcoun{F}})_{((X,\vec{A}\rstn_n),A_{n+1})} = (\fb{\Jcoun{F}})_{(X,\vec{A}\rstn_n)} \arcmp \jcomp{q}\ft{F}A_{n+1}^*$ since it is image of $(\ft{\Jcoun{F}})_{((X,\vec{A}\rstn_n),A_{n+1})}$ under $\jcomp{q}$.
			Finally, we have $d^0_n\arcmp(\fb{\Jcoun{F}})_{(X,\vec{A}\rstn_n)}=\fb{F}\C{0}{n}$ by inductive hypothesis. These equalities let us conclude that the diagram below commutes.
			\[\begin{tikzcd}
				{\fb{F}\comp{p}A_n^*} & {\fb{F}\comp{p}A_{n+1}^*} & {\comp{q}\ft{F}A_{n+1}^*} \\
				{\comp{q}(\ft{F}A_n)^*} & {\fb{F}X} & {\comp{q}(\ft{F}A_{n+1})^*}
				\arrow["{(\fb{\Jcoun{F}})_{(X,\vec{A}\rstn_n)}}"', from=1-1, to=2-1]
				\arrow["{\fb{F}\C{0}{n}}"{description}, from=1-1, to=2-2]
				\arrow["{\fb{F}\jcomp{p}A_{n+1}^*}"', from=1-2, to=1-1]
				\arrow["{\psi_{A_{n+1}^*}}", from=1-2, to=1-3]
				\arrow["{\fb{F}\C{0}{n+1}}", from=1-2, to=2-2]
				\arrow["{\jcomp{q}\ft{F}A_{n+1}^*}"', bend right, from=1-3, to=1-1]
				\arrow["{\comp{q}(\ft{\Jcoun{F}})_{((X,\vec{A}\rstn_n),A_{n+1})}}", from=1-3, to=2-3]
				\arrow["{d^0_n}"', from=2-1, to=2-2]
				\arrow["{\jcomp{q}(\ft{F}A_{n+1})^*}", bend left, from=2-3, to=2-1]
				\arrow["{d^0_{n+1}}", from=2-3, to=2-2]
			\end{tikzcd}\]

\begin{rem}
	The construction of $\Jcoun{F}$ shows that the naturality square of $\Jcoun{p}$ does not necessarily commute even if $F$ preserves comprehension on the nose. In fact, in order to prove that $\Jcoun{F}$ is the identity, one would also need that $F$ strictly preserves the cleavage.
\end{rem}

We are finally ready to prove the main result of this section.			
			
			\begin{thm}\label{thm:jfree}
				The 2-functor $\JFree$ is left bi-adjoint to the 2-functor $\forgJ$.
			\end{thm}
			
			\begin{proof}
				The pseudo-naturality of unit and counit are proved in \zcref{prop:junit} and \zcref{prop:jcounit}. Triangular identities are shown in \zcref{prop:trian-id-1} and \zcref{prop:trian-id-2}.
			\end{proof}

\section{From Jacobs to Lawvere-Ehrhard comprehension categories}
\label{sect:JtoL} 

Our next goal is to describe a universal construction turning any comprehension category into a Lawvere-Ehrhard one. 
Taking inspiration from the characterization we have proved in \zcref{sect:lcomp-vs-jcomp}, 
we will do this in two steps. 
First, we will show how to freely add fibred terminal objects to a comprehension category in such a way that they are preserved by the comprehension functor, i.e., their comprehension is an isomorphism. 
This corresponds to extending a dependent type theory with a unit type in every context inhabited by a unique term. 
Then, we will describe how to turn a comprehension category with fibred terminal objects preserved by the comprehension functor into a Lawvere-Ehrhard comprehension category. This construction is the one that forces the isomorphism between terms and global type morphisms that characterizes Lawvere-Ehrhard comprehension categories. 
As in the previous section, these free constructions are first defined at the level of objects, then extended to the considered 2-categories and finally proved to be left bi-adjoint using the triangular identities. In particular, we will use the fact that the forgetful functor $\LtoJ:\LComp\to\JComp$ factors through the forgetful functor $\forgT:\JTComp\to\JComp$. To improve readability, auxiliary technical lemmas are postponed to \zcref{app:B}. 

Let us start by introducing comprehension categories with terminals.

\begin{defi}
	A \define{comprehension category with terminals} is a comprehension category whose underlying fibration $p$ has fibred terminal objects and these are preserved by $\jcomp{}$.
\end{defi}

This means that a comprehension category with terminals is such that the natural transformation $\jcomp{}\term{}:\comp{}\term{}\nt \id{\ct{B}}$, given by the comprehension of fibred terminal objects, is invertible.

Comprehension categories with terminals can be organized into a 2-category $\JTComp$ which is the 2-full 2-subcategory of $\JComp$ on morphisms which preserve fibred terminal objects. Hence, there is a forgetful 2-functor $\forgT:\JTComp\to\JComp$ which forgets the presence of fibred terminal objects.

Note that by \zcref{prop:comp-term-inv} the 2-functor $\LtoJ:\LComp\to\JComp$ maps a LECC to a comprehension category with terminals. It also maps a 1-cell $F$ of $\LComp$ to a 1-cell of $\JTComp$, since $F$ preserves fibered terminal objects when seen as a morphism of fibrations. Then it maps 2-cells of $\LComp$ to 2-cells of $\JTComp$, since the latter is 2-full in $\JComp$. 
These facts imply that $\LtoJ:\LComp\to\JComp$ factors through $\forgT:\JTComp\to\JComp$, giving rise to a forgetful 2-functor $\forgLE:\LComp\to\JTComp$.
Then,  in the rest of this section, we will describe the left bi-adjoints of 
$\forgT$ and $\forgLE$, 
thus obtaining by composition a left bi-adjoint of $\LtoJ$.

\subsection{The free comprehension category with terminals}

The process that freely adds fibred terminals to a comprehension category is easy to understand.
In fact, it is enough to add a formal terminal object to each fibre, and set its comprehension equal to the identity.
More precisely, given a comprehension category $p:\ct{E}\to\ct{B}$, 
we first consider the category $\JTDomE{p}$  obtained by adding to $\ct{E}$ an object $\term{}X$  for every object $X$ in $\ct{B}$ and 
by setting $\Hom{\JTDomE{p}}(\term{}X,\term{}Y) = \Hom{\ct{B}}(X,Y)$  and $\Hom{\JTDomE{p}}(A, \term{}X) = \Hom{\ct{B}}(pA, X)$, where $A$ ranges over objects of $\ct{E}$.
Note that morphisms with domain $\term{}X$ exist only if the codomain is of the form $\term{}Y$.
Equivalently, arrows into an object from $\ct{E}$ must come from $\ct{E}$ themselves.
Composition and identities in $\JTDomE{p}$ are defined in the obvious way, notably, 
if $f : A\to B$ is an arrow from $\ct{E}$ and $g : B \to \term{}X$ is a new arrow, we have 
$g\arcmp f = g\arcmp p(f)$, i.e., we use the composition of $\ct{B}$. 
Then, we consider the functor $\JTFib{p} : \JTDomE{p}\to\ct{B}$ defined as follows: 
on objects and arrows from $\ct{E}$ it acts as $p$, 
it maps an object $\term{}X$ to $X$ and  a morphism $f \in \Hom{\JTDomE{p}}(Q,\term{} X) = \Hom{\ct{B}}(\JTFib{p}(Q),X)$ to itself. 
Notice also that the assignment $X\mapsto\term{} X$ trivially extends to a functor 
$\term{} : \ct{B} \to \JTDomE{p}$. 

\begin{prop}
The functor $\JTFib{p} : \JTDomE{p} \to \ct{B}$ is a fibration and $\term{} : \ct{B} \to \JTDomE{p}$ is a fibred terminal object functor. 
\end{prop}
\begin{proof}
Consider $f:X\to Y$ in $\ct{B}$. 
Cartesian liftings of $f$ at objects $A$ from $\ct{E}$ are exactly those of $p$, as the only arrows in $\JTDomE{p}$ into $A$ are those coming from $\ct{E}$. 
We now show that the cartesian lifting of $f$ at $\term{} Y$ is $f$ itself. 
In fact, consider $g$ in a $\Hom{\JTDomE{p}}(Q,\term{}Y) = \Hom{\ct{B}}(\JTFib{p}Q,Y)$ and $h:\JTFib{p}Q\to X$ such that $f\arcmp h = g$. 
By definition of $\JTDomE{p}$ and $\JTFib{p}$, the only morphism over $h$ with respect to $\JTFib{p}$ is $h$ itself, which trivially makes  the required triangle commute. 

To check that $\term{}$ is a fibred terminal object functor, it suffices to verify that it is a right adjoint right inverse of $\JTFib{p}$, but this is trivial because we have $\JTFib{p}(\term{} X) = X$ and $\Hom{\JTDomE{p}}(Q,\term{}X) = \Hom{\ct{B}}(\JTFib{p}Q,X)$. 
\end{proof} 

We can also define 
a functor $\jcomp{\JTFib{p}}:\JTDomE{p}\to\Arr{\ct{B}}$ as follows: 
on objects and arrows from $\ct{E}$ it acts as $\jcomp{p}$, 
it maps an object $\term{}X$ to the identity $\id{X} : X \to X$, and 
it maps an arrow $f \in \Hom{\JTDomE{p}}(Q,\term{}X) = \Hom{\ct{B}}(\JTFib{p}Q,X)$ to the commutative square 
\[\begin{tikzcd}
  \comp{\JTFib{p}}Q 
  \ar[r, "f\arcmp\jcomp{\JTFib{p}}A"]
  \ar[d, "\jcomp{\JTFib{p}}A"']
& X 
  \ar[d, "\id{X}"]
\\
  \JTFib{p}Q 
  \ar[r, "f"'] 
& X 
\end{tikzcd}\]

\begin{prop}
$(\JTFib{p},\jcomp{\JTFib{p}})$ is a comprehension category with terminals. 
\end{prop}
\begin{proof}
Since we already know that $\jcomp{\JTFib{p}}$ maps fibred terminal objects to isomorphisms, it suffices to check that it maps cartesian arrows to pullbacks. 
This is trivial for cartesian arrows coming from $\ct{E}$ as on those $\jcomp{\JTFib{p}}$ acts as $\jcomp{p}$. 
The only other cartesian arrows have shape $f : \term{}X \to \term{}Y$, which are mapped to the square 
\[\begin{tikzcd}
  X \ar[r, "f"]
    \ar[d, "\id{X}"']
& Y \ar[d,"\id{Y}"]
\\
  X \ar[r, "f"']
& Y
\end{tikzcd}\]
that is obviously a pullback. 
\end{proof}

The construction of the comprehension category with terminals $(\JTFib{p},\jcomp{\JTFib{p}})$
extends to a 2-functor $\JTFree:\JComp\to\JTComp$ in the following way. 
Given a 1-cell $F:p\to q$, its image $\JTFun{F}$ is such that 
$\fb{\JTFun{F}} = \fb{F}$ and 
$\ft{\JTFun{F}}$ acts like $\ft{F}$ on objects and arrows from the total category of $p$, 
it maps $\term{p}X$ to $\term{q}\fb{F}X$, and 
a morphism $f \in\Hom{\JTDomE{p}}(Q,\term{p}X)$ to $\fb{F}f \in \Hom{\JTDom{q}}(\ft{\JTFun{F}}Q, \term{q}\fb{F}X)$. 
Given a 2-cell $\alpha$, its image $\JTFun\alpha$ is defined by 
$\fb{\JTFun\alpha}=\fb\alpha$, 
$\ft{\JTFun\alpha}_A = \ft\alpha_A$, for $A$ an object from original total category of the domain fibration of $\alpha$, and 
$\ft{\JTFun\alpha}_{\term{}X} = \fb\alpha_X$, for $X$ an object in the base of the domain fibration of $\alpha$. 

In order to prove that $\JTFree$ is universal, we define two pseudo-natural transformation that will constitute the unit and counit of the biadjunction $\JTFree\dashv \forgT$.

			The unit is easy to describe: given a comprehension category $p$, the morphism $\JTun{p}:p\to \forgT \JFib{p}$ is given by the inclusion on the total category, and the identity on the base. This morphism preserves comprehension on the nose: 
			
The counit requires slightly more attention: 
comprehensions will not be preserved strictly in general, but they will be preserved only up to iso. 
This fact is a consequence of having set the comprehension of fibred terminal objects equal to the identity, while for an arbitrary comprehension category with terminals the comprehension of fibred terminal objects is just required to be invertible.
Given a comprehension category with terminals $p:\ct{E}\to\ct{B}$, the morphism $\JTcoun{p}:\JTFib{p}\to p$ is defined by setting $\fb{\JTcoun{p}}\colon = \Id{\ct{B}}$ and $\ft{\JTcoun{p}}$ does nothing on objects and morphisms coming from $\ct{E}$, and maps new morphisms $f\in\Hom{\JTDomE{p}}(A,\term{}X)=\Hom{\ct{B}}(\JTFib{p}A,X)$ to the unique morphism $!:A\to \term{p}X$ over $f:pA\to X$. 
If $F:p\to q$ is a morphism in $\JTComp$, we define $\fb{\JTcoun{F}}\colon = \Id{\ct{B}}$ and $\ft{\JTcoun{F}}:\ft{F}\arcmp\ft{\JTcoun{p}}\nt \ft{\JTcoun{q}}\arcmp \ft{\JFun{F}}$ by using the universal property of fibred terminal objects.
			
			Using the unit and the counit just defined, we can prove the main theorem of this section.

			\begin{thm}
				The 2-functor $\JTFree:\JComp\to\JTComp$ is left bi-adjoint to the forgetful functor $\forgT:\JTComp\to\JComp$.
			\end{thm}
			
			\begin{proof}
				The unit and counit of the bi-adjunction are, respectively, $\JTun{}$ and $\JTcoun{}$. They are pseudo-natural by \zcref{lem:JTun-pseudonat} and \zcref{lem:JTcoun-pseudonat}. Both triangular identities are trivial since $\ft{\JTun{p}}$ is the inclusion and $\ft{\JTcoun{p}}$ does nothing on elements of $\ct{E}$, while the bottom components $\fb{\JTun{p}}$ and $\fb{\JTcoun{p}}$ are the identity.
			\end{proof}

			\subsection{The free Lawvere-Ehrhard comprehension category}

We now describe the free Lawvere-Ehrhard comprehension category over a comprehension category with terminals. In this case, the construction is obtained by using the universal property of a 2-coequalizer. This will guarantee that the unit and the counit are actually 2-natural and that triangular identities are strict, giving rise to a 2-adjunction.

The construction can be sketched as follows: given a comprehension category with terminals $p:\ct{E}\to\ct{B}$, we want to force it to satisfy the conditions of \zcref{thm:ess-im}, that is, to have an isomorphism between terms and global type morphisms.
By \zcref{lem:term-pres} 
we know that there is a functor mapping global type morphisms to terms. 
Following \zcref{thm:ess-im}, we need a way to turn it into an isomorphism.
This requires to solve two problems: 
on the one hand, we have to add ``missing'' global type morphisms and, on the other, 
we have to identify those that correspond to the same term. 
Indeed, recall that two global type morphisms $f,g : \term{}X \to B$ correspond to the same term if 
$\comp{}f\arcmp \eta_X = \comp{}g\arcmp \eta_X$, where $\eta_X : X \to \comp{}\term{}X$ is the inverse of $\jcomp{}\term{}X$.

We start by adding to $\ct{E}$ morphisms corresponding to terms, i.e., sections of comprehension arrows, closing them under composition. In particular we consider the category $\Tr{p}$ 
with the same objects as $\ct{E}$ and where morphisms from $A$ to $B$ 
are either maps $f:A\to B$ in $\ct{E}$ or maps $t:pA\to \comp{}B$ in $\ct{B}$. 
The composition is given as follows: 
for two morphisms from $\ct{E}$ it is the same as in $\ct{E}$, 
for two morphisms $t : pA\to\comp{}B$ and $s : pB\to\comp{}C$ it  is given by $s\arcmp\jcomp{}B\arcmp t$, 
for morphisms $f:A\to B$ and $s:pB\to \comp{}C$ it is given by $s\arcmp pf$, and, finally, 
for morphisms $t:pA\to \comp{}B$ and $g:B\to C$ it is given by $\comp{}g\arcmp t$.
It is easy to verify that this composition  is associative and has identities given by those of $\ct{E}$. The following diagram shows all three non trivial possible compositions.

\[\begin{tikzcd}
	{\comp{}A} && {\comp{}B} && {\comp{}C} & \\
	\\
	& pA && pB && pC
	\arrow["{\jcomp{}A}"', from=1-1, to=3-2]
	\arrow["{\comp{}g}", from=1-3, to=1-5]
	\arrow["{\jcomp{}B}"', from=1-3, to=3-4]
	\arrow["{\jcomp{}C}", from=1-5, to=3-6]
	\arrow["t", from=3-2, to=1-3]
	\arrow["pf"', from=3-2, to=3-4]
	\arrow["s"', from=3-4, to=1-5]
\end{tikzcd}\]

Clearly, there is a functor $\Trfib{p}$ from $\Tr{p}$ into $\ct{B}$, which acts as $p$ on objects and on morphisms from $\ct{E}$ and 
it maps a morphism $t:pA\to\comp{}B$ to $\jcomp{}B\arcmp t$ 
(it postcomposes $t$ with the comprehension of $B$). Its functoriality is immediate looking at the previous diagram.
Moreover, such a functor is a fibration: 
the cartesian lifting $\carar{f}{A}$ of $f$ at $A$ is the same as the one with respect to $p$. 
The factorization is given by cartesianity for arrows coming from $p$, and by the universal property of the pullback $\jcomp{}\carar{f}{A}$ for morphisms of the new form. 
Moreover, it is clear that $\ct{E}$ is a (non-full) subcategory of $\Tr{p}$.

Notice that we can endow $\Trfib{p}$ with a comprehension structure which makes it into a comprehension category. In fact, consider the functor $\jcomp{\Trfib{p}}:\Tr{p}\to \Arr{\ct{B}}$ defined as follows: it acts as $\jcomp{}$ on objects and morphisms from $\ct{E}$, and it maps a morphism $t:pA\to\comp{}B$ to the square
\[\begin{tikzcd}
	{\comp{}A} && {\comp{}B} \\
	pA && pB
	\arrow["{t\arcmp\jcomp{}A}", from=1-1, to=1-3]
	\arrow["{\jcomp{}A}"', from=1-1, to=2-1]
	\arrow["{\jcomp{}B}", from=1-3, to=2-3]
	\arrow["t"{description}, from=2-1, to=1-3]
	\arrow["{\jcomp{}B\arcmp t}"', from=2-1, to=2-3]
\end{tikzcd}\]
It is straightforward to check that it is a functor. It maps cartesian morphisms to pullbacks since so does $\jcomp{}$ and 
cartesian morphisms of $\Trfib{p}$ and of $p$ coincide.

At this point, we have constructed a comprehension category where every term ``comes from'' a corresponding type morphism. 
However, such a type morphism is not necessarily global as this comprehension category has not fibred terminal objects in general and, moreover, 
there can be more than one type morphism inducing the same term. 
All these issues are essentially due to the fact that the total category of the fibration now has too many morphisms.
To solve this problem,  we need to identify those type morphisms in $\Tr{p}$ that should be global and correspond to the same term and, 
to obtain a category, we also need to close this identification under composition. 
Altogether, this means that we have to identify morphisms in $\Tr{p}$ which correspond to the same constant morphism in $\ct{E}$, that is, a morphism that factors through a fibred terminal object. 

To this end, we will compute a coequalizer of two functors into $\Tr{p}$ from a category of constant morphisms in $\ct{E}$. 
We would like to consider 
a category $\Const{p}$ where objects are the same as $\ct{E}$ and 
morphisms from $A$ to $B$ are morphisms $f : \term{} pA\to B$ in $\ct{E}$. 
Two arrows $f : \term{}pA\to B$ and $g : \term{}pB \to C$ can be composed as in  the diagram below.

			\[\begin{tikzcd}
				{\term{}pA} & B & {\term{}pB} & C
				\arrow["f", from=1-1, to=1-2]
				\arrow["{!}", from=1-2, to=1-3]
				\arrow["g", from=1-3, to=1-4]
			\end{tikzcd}\]

It is easy to see that this composition is associative. 
However, it has no identities, intuitively because identity arrows in $\ct{E}$ are not constant. 
Hence, the structure that we get is actually weaker than a category, and it is called semicategory.

We recall from \cite[Sec. ~4]{mitchell1972dominion}
that a \define{semicategory} is a category without identities, and a semifunctor is an assignment mapping objects to objects and morphisms to morphisms, and 
preserving composition. 
Moreover, for every semicategory we can consider the free category over it, obtained by adding formal identities to every object. 

We now consider the free category over the semicategory $\Const{p}$, denoted by $\ConstCat{p}$. 
Observe that functors from $\ConstCat{p}$ to a category $\ct{C}$ are the same as semifunctors from $\Const{p}$ to $\ct{C}$ (regarded as a semicategory). 
By applying this fact to the restriction of $\ct{E}$ to $\Const(p)$, we get a functor 
$\ConstFib{p}:\ConstCat{p}\to\ct{B}$. 
However, in general it is not a fibration, because it lacks cartesian liftings.

We now define the two functors $\ConstCat{p}\to\Tr{p}$ that we will then coequalize to obtain the final construction. 
We do this by transposing two semifunctors
$\incl,\triang:\Const{p}\to\Tr{p}$. 
These are the identity on objects and, 
given $f\in \Hom{\Const{p}}(A,B) = \Hom{\ct{E}}(\term{}pA, B)$, 
 we set $\incl f\colon=f\arcmp !$ and $\triang f\colon=\comp{} f\arcmp \eta_X$. 
With a slight abuse of notation, we will identify them with 
the unique functors from $\ConstCat$ they induce via transposition. 

We are now ready to define the functor that will be the free Lawvere-Ehrhard comprehension category. 

			\begin{defi}
				The functor $\LEFib{p}:\LEDomE{p}\to\ct{B}$ is the coequalizer of $\incl$ and $\triang$, seen as 1-cells from $\ConstCat{p}$ to $\ConstFib{p}$ in $\Cattwo/\ct{B}$.
				\[\begin{tikzcd}
					{\ConstCat{p}} && {\Tr{p}} && {\LEDomE{p}} \\
					{\ct{B}} && {\ct{B}} && {\ct{B}}
					\arrow["\triang"', shift right=2, from=1-1, to=1-3]
					\arrow["\incl", shift left=2, from=1-1, to=1-3]
					\arrow["{\ConstFib{p}}"', from=1-1, to=2-1]
					\arrow["{\quot{p}}"{description}, from=1-3, to=1-5]
					\arrow["{\Trfib{p}}"{description}, from=1-3, to=2-3]
					\arrow["{\LEFib{p}}", from=1-5, to=2-5]
					\arrow["{\Id{\ct{B}}}"', from=2-1, to=2-3]
					\arrow["{\Id{\ct{B}}}"{description}, from=2-3, to=2-5]
				\end{tikzcd}\]
			\end{defi}

Since coequalizers in $\Cattwo/\ct{B}$ are computed as in $\Cattwo$, 
we know that $\LEDomE{p}$ is itself the coequalizer of  $\incl$ and $\triang$ regarded as 1-cells in $\Cattwo$.

From \cite[Prop 3.16 and Prop. 4.1]{bednarczyk1999generalized} 
it follows that the coequalizer is computed by quotienting morphisms of $\Tr{p}$ 
by the arrow congruence $\sim$ generated by the relation $R$ such that 
 $f R g$ if and only if there exists $h$ in $\ConstCat{p}$ such that $\incl h=f$ and $\triang h=g$. 
In other words, $\sim$ is the smallest equivalence relation that extends $R$ and it is stable under composition, that is, 
if $f\sim f'$ and $g\sim g'$ then $g\arcmp f \sim g' \arcmp f'$.

			\begin{prop}
				The functor $\LEFib{p}$ is a fibration.
			\end{prop}
			
			\begin{proof}
				We claim that cartesian liftings are obtained by applying $\quot{}$ to the correspondent cartesian liftings in $\Tr{p}$. In fact, given $f:A^*\to A$ cartesian with respect to $\Trfib{p}$, consider $h:B\to A$ in $\Tr{p}$ and $g:Z\to X$ in $\ct{B}$ such that $\Trfib{p}h=\Trfib{p}f\arcmp g$. We need to show that there is a unique $r:B\to A^*$ in $\LEDomE{p}$ over $g$ such that $\quot{}h=\quot{}f\arcmp r$. Equivalently, we want to find a $t:B\to A^*$ in $\Tr{p}$ over $g$ such that $\quot{}h=\quot{}f\arcmp \quot{}t$, and for any $t'$ over $g$ which satisfies the same equation we want to have $t\sim t'$.
				
The existence of $t$ is granted by the cartesianity of $f$ with respect to $\Trfib{p}$. 
So we only need to show the following: given $h':B\to A$ in $\Tr{p}$ which factorizes along $f$ in $h'=f\arcmp t'$, if $h\sim h'$ then $t\sim t'$. We do this by induction on the definition of $\sim$.
				
				\begin{itemize}
					\item if $h=h'$, then $t=t'$ by cartesianity of $f$, hence $t\sim t'$;
					\item if $h\sim h''$ and $h''\sim h'$, then by cartesianity we get $t''$ such that $h''=f\arcmp t''$. By inductive hypothesis $t\sim t''$ and $t''\sim t'$, hence $t\sim t'$ by transitivity;
					\item if $h=a\arcmp b$ and $h'=a'\arcmp b'$ with $a\sim a'$ and $b\sim b'$, then we can factor $a$ and $a'$ in $a=f\arcmp c$ and $a'=f\arcmp c'$ by cartesianity. Notice that $\Trfib{p}a=\Trfib{p}a'$ since they are in the same equivalence class, hence we are allowed to use the cartesianity of $f$. Finally we have $c\sim c'$ by inductive hypothesis, hence $t=c\arcmp b\sim c'\arcmp b'=t'$, where the two equalities hold again by cartesianity of $f$;
					\item if there is $s$ in $\ConstCat{p}$ such that $\incl s=h$ and $\triang s= h'$, then by cartesianity there is $c:\term{}Z\to\reind{A}$ such that $s=f\arcmp c$. Furthermore $t=c\arcmp !$, by unicity of the cartesian factorization. Then using the universal property of the pullback it is clear that $\comp{}c\arcmp \eta_Z=t'$. These two facts mean that $\incl c=t$ and $\triang c=t'$, hence $t\sim t'$;
					\[\begin{tikzcd}
						{\comp{}\term{}Z} &&&&&&&&& \\
						&& {\comp{}\reind{A}} && {\comp{}A} && {\term{}Z} \\
						& Z &&&&& B \\
						&& X && Y &&& {\reind{A}} && A
						\arrow["{\comp{}c}"{description}, from=1-1, to=2-3]
						\arrow["{\comp{}s}"{description}, bend left, from=1-1, to=2-5]
						\arrow["{\jcomp{}\term{}Z}"', bend right, from=1-1, to=3-2]
						\arrow["{\comp{}f}", from=2-3, to=2-5]
						\arrow["{\jcomp{}\reind{A}}"{description}, from=2-3, to=4-3]
						\arrow["{\jcomp{}A}"{description}, from=2-5, to=4-5]
						\arrow["c"{description}, dashed, from=2-7, to=4-8]
						\arrow["s", from=2-7, to=4-10]
						\arrow["{\eta_Z}"{description}, from=3-2, to=1-1]
						\arrow["{t'}", from=3-2, to=2-3]
						\arrow["g"', from=3-2, to=4-3]
						\arrow["{!}", from=3-7, to=2-7]
						\arrow["t"', from=3-7, to=4-8]
						\arrow["h"{description}, from=3-7, to=4-10, crossing over]
						\arrow["pf"', from=4-3, to=4-5]
						\arrow["f"', from=4-8, to=4-10]
					\end{tikzcd}\]
					\item if there is $s$ in $\ConstCat{p}$ such that $\triang s= h$ and $\incl s=h'$, then it is analogous to the precedent case.
				\end{itemize}
			\end{proof}

			The fibration $\LEFib{p}$ can be given a comprehension structure which turns it to a JCC. In particular, the comprehension functor $\jcomp{\Trfib{p}}:\Tr{p}\to\Arr{\ct{B}}$ induces a functor $\jcomp{\LEFib{p}}:\LEDomE{p}\to\Arr{\ct{B}}$ by using the universal property of the coequalizer. It is clear in fact that $\jcomp{\Trfib{p}}$ coequalizes $\incl$ and $\triang$.

			\begin{prop}
				The pair $(\LEFib{p}, \jcomp{\LEFib{p}})$ is a comprehension category with terminals.
			\end{prop}
			
			\begin{proof}
				First we need to show that $\LEFib{p}$ has fibred terminal objects. We claim that they are the same as fibred terminal objects of $p$. There is always the map $!:A\to \term{}X$ in $\Tr{p}$. We need to show that if $f:A\to\term{}X$ is vertical in $\Tr(p)$, then $f\sim !$. If $f$ is an arrow in $\ct{E}$, then $f$ is equal to $!$. Otherwise, $f:X\to \comp{}\term{}X$ is equal to $\eta_X$ since $\jcomp{}\arcmp f=\id{X}$. In both cases $f\sim !$.
				
				We already know that $\jcomp{\LEFib{p}}:\LEDomE{p}\to\Arr{\ct{B}}$ is a functor. The only thing left to prove is that it maps cartesian morphisms to pullbacks and preserves fibred terminal objects. The former condition is trivial since cartesian morphisms in $\LEDomE{p}$ are equivalence classes of cartesian morphisms in $\ct{E}$, and $\jcomp{\LEFib{p}}$ acts on them in the same way as $\jcomp{}$.
				The preservation of fibred terminal objects is trivial since $\jcomp{\LEFib{p}}$ coincides with $\jcomp{}$ on objects, and the latter preserves fibred terminal objects by assumption.
			\end{proof}
			
			We now use the characterization given in \zcref{sect:lcomp-vs-jcomp} to finally show that this construction yields indeed a LECC.
			
			\begin{prop}
				$\LEFib{p}$ is a Lawvere-Ehrhard comprehension category. 
			\end{prop}
			
			\begin{proof}
				We want to use the characterization of \zcref{lem:ess-im}. The first condition is satisfied since $\jcomp{\LEFib{p}}$ coincides with $\jcomp{}$ on objects, and $\jcomp{}$ preserves terminals by assumption.
				
				Consider now $A$ over $X$ and $t:X\to\comp{}A$ a section of $\jcomp{\LEFib{p}}A$. We need to show that there is a unique vertical morphism $t^{\#}:\term{}X\to A$ such that $t=\comp{\LEFib{p}}t^{\#}\arcmp \eta_X$. First of all, one can notice that $t$ is a morphism from $\term{}X$ to $A$ in $\Tr(p)$, hence we define $t^{\#}\colon= \quot{} t$. It is clear that $t=\comp{\LEFib{p}}t^{\#}\arcmp\eta_X$, since $\comp{\LEFib{p}}t^{\#}=\comp{\Trfib{p}}t=t\arcmp\jcomp{}\term{}X$. For the unicity, suppose that $s:\term{}X\to A$ in $\Tr{p}$ is such that $\quot{}s$ satisfies the same equation. Since the domain of $s$ is $\term{}X$, there is a morphism $h$ in $\ConstCat{p}$ which is mapped to $s$ itself by either $\incl$ or $\triang$. In the former case, $s\sim \triang h$ by the base case of the definition of $\sim$, while in the latter it holds by reflexivity. This means that in both cases the following holds:
				\[s\sim \triang h=\comp{}h\arcmp\eta_X=\comp{\Trfib{p}}s\arcmp\eta_X=\comp{\LEFib{p}}\quot{}s\arcmp\eta_X=\comp{\LEFib{p}}\quot{}t\arcmp\eta_X=\comp{}t\arcmp\eta_X\sim t\]
				
				To prove the naturality, consider $f:A\to B$ in $\LEDomE{p}$, $t:X\to\comp{}A$ and $s:Y\to\comp{}B$ such that $s\arcmp \LEFib{p}f=\comp{\LEFib{p}}f\arcmp t$. An easy calculation shows that $s^{\#}\arcmp \term{\LEFib{p}}\LEFib{p}f=f\arcmp t^{\#}$.
			\end{proof}
			
			Now we proceed in proving that this construction is not only 2-functorial, but also provides a left 2-adjoint to the forgetful functor $\forgLE:\LComp\to\JTComp$. 
			
			Consider a morphism $F:p\to q$ in $\JTComp$, with $p:\ct{E}\to\ct{B}$ and $q:\ct{E'}\to\ct{B'}$. We define the functor $\Tr{F}:\Tr{p}\to\Tr{q}$ which acts as $\ft{F}$ on objects and morphisms from $\ct{E}$, and it maps arrows $t:pA\to\comp{p}B$ to $\Tr{F}t\colon = \alpha^{-1}_B\arcmp\fb{F}t$ where $\alpha:\comp{q}\ft{F}\to\fb{F}\comp{p}$ is the natural isomorphism preserving comprehension. This functor induces a morphism of comprehension categories $\Trfib{F}:\Trfib{p}\to\Trfib{q}$. First of all, it is a fibration morphism since every cartesian morphism of $\Trfib{p}$ comes from $\ct{E}$, and $\ft{F}$ preserves them.
			The preservation of comprehension of objects and morphisms coming from $\ct{E}$ is a consequence of $F$ being a morphism of comprehension categories. If instead $f\in\Hom{\ct{B}}(pA,\comp{p}B)\subseteq\Hom{\Tr{p}}(A,B)$, we have that $\fb{F}\comp{\Tr{p}}f=\fb{F}f\arcmp \fb{F}\jcomp{p}A$ and $\comp{\Tr{q}}\Tr{F}f=\alpha^{-1}_B\arcmp \fb{F}f\arcmp \jcomp{q}\ft{F}A$. These are again isomorphic since $F$ preserves comprehension.
			
			We also define the functor $\ConstCat{F}:\ConstCat{p}\to\ConstCat{q}$ which acts as $\ft{F}$ on objects and maps a morphism $f:\term{p}X\to B$ to $\ConstCat{F}f\colon = \ft{F}f\arcmp \beta^{-1}_X$, where $\beta:\term{q}\ft{F}\to\ft{F}\term{p}$ is the natural isomorphism preserving fibred terminal objects.
			
			Consider now a 2-cell $\gamma:F\to G$ in $\JTComp$. We can use the universal property of the 2-coequalizer to define a 2-cell $\LEFun{\gamma}:\LEFun{F}\to\LEFun{G}$. This gives rise to a 2-functor $\LEFree:\LComp\to\JTComp$.
			
			We are now ready to define the unit and counit of the 2-adjunction $\LEFree\dashv\forgLE$. First, take a comprehension category $p$ with terminals. The morphism $\LEun{p}:p\to\LEFib{p}$ is specified by the composition $\ft{\LEun{p}}:\ct{E}\hookrightarrow\Tr{p}\to\LEDomE{p}$. It clearly preserves terminals and comprehensions.
			
			Consider now a LECC $p$. Notice that in this case a morphism from $A$ to $B$ in $\Tr{p}$ of the form $f:X\to\comp{}A$ corresponds via transposition to a morphism from $A$ to $B$ in $\Const$. The component at $p$ of the counit, $\LEcoun{p}$, is determined by the functor $\Tr{p}\to \ct{E}$ which is the identity on objects, and maps morphisms of the form $f:X\to\comp{}A$ to $\incl (f^{\#})$. These allow us to conclude with the main result of the section.
			
			\begin{thm}
				The 2-functor $\LEFree:\JTComp\to\LComp$ is a left 2-adjoint to the forgetful 2-functor $\forgLE:\LComp\to \JTComp$.
			\end{thm}
			
			\begin{proof}
				The unit and the counit of the 2-adjunction are respectively $\LEun{}$ and $\LEcoun{}$. They are 2-natural by \zcref{lem:LEun-nat} and \zcref{lem:LEcoun-nat}. Triangular identities follow directly from the definition of $\LEcoun{p}$.
			\end{proof}

			\section{Conclusions}
			\label{sect:conclu} 

In this paper, we have systematically investigated the relationship between two distinct categorical models of type dependency: Jacobs comprehension categories and Lawvere-Ehrhard comprehension categories. 
By comparing their respective fibrations of terms and type morphisms from the unit type, we have identified the structural principle that distinguishes them: 
whereas Jacobs comprehension categories treat type morphisms as independent data, 
in Lawvere-Ehrhard comprehension categories terms are completely determined by type morphisms from the unit type,  that is, 
 these two fibrations are isomorphic. 
Then, we have described three free constructions relating Jacobs and Lawvere-Ehrhard comprehension categories with each other and with plain fibrations. 
The obtained results are summarized in the following diagram of 2-categories: 

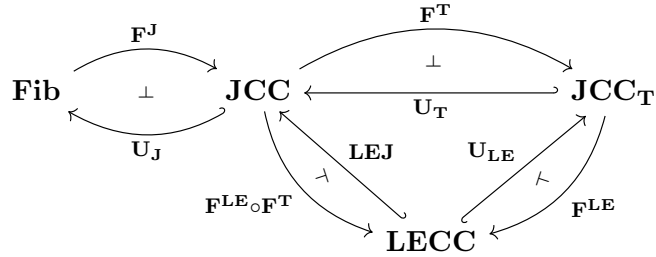
\[\begin{tikzcd}
	\Fib && \JComp && \JTComp \\
	\\
	&&& \LComp
	\arrow[""{name=0, anchor=center, inner sep=0}, "\JFree", bend left, from=1-1, to=1-3]
	\arrow[""{name=1, anchor=center, inner sep=0}, "\forgJ", bend left, hook', from=1-3, to=1-1]
	\arrow[""{name=2, anchor=center, inner sep=0}, "\JTFree", bend left, from=1-3, to=1-5]
	\arrow[""{name=3, anchor=center, inner sep=0}, "{\LtoJ}"', from=3-4, to=1-3, hook']
	\arrow[""{name=4, anchor=center, inner sep=0}, "\forgT", hook', from=1-5, to=1-3]
	\arrow[""{name=5, anchor=center, inner sep=0}, "\LEFree", bend left, from=1-5, to=3-4]
	\arrow[""{name=6, anchor=center, inner sep=0}, "\LEFree\arcmp\JTFree"', bend right, from=1-3, to=3-4]
	\arrow[""{name=7, anchor=center, inner sep=0}, "\forgLE", hook', from=3-4, to=1-5]
	\arrow["\dashv"{anchor=center, rotate=-90}, draw=none, from=0, to=1]
	\arrow["\dashv"{anchor=center, rotate=-90}, draw=none, from=2, to=4]
	\arrow["\dashv"{anchor=center, rotate=162}, draw=none, from=5, to=7]
	\arrow["\dashv"{anchor=center, rotate=18}, draw=none, from=6, to=3]
\end{tikzcd}\]

The 2-functor $\JFree$ builds the free comprehension category over a fibration, 
the 2-functor $\JTFree$ turns a comprehension category into a comprehension category with terminals, that is, it freely adds fibred terminal objects preserved by the comprehension functor, and 
$\LEFree$ returns the free Lawvere-Ehrhard comprehension category over a comprehension category with terminals. 
Furthermore, the compositions $\LEFree\circ\JTFree$ and $\LEFree\circ\JTFree\circ\JFree$  give the free Lawvere-Ehrhard comprehension category  over a comprehension category and over a fibration, respectively.

These free constructions capture in a principled categorical way the precise syntactic features that differentiate these models from one another. 
Moreover, they provide us with tools for modularly extending  models with new features and for generating a wide variety of new, free examples. 
Altogether, our analysis highlights 
how the presence of non-trivial type morphisms yields a richer interaction among components of  a dependent type theory, and 
how the notion of Lawvere-Ehrhard comprehension tames this interplay.

\subsubsection*{Related work}
In the last few years, there is a growing interest in the study of type theories and their models supporting some form of type morphisms. 
Coraglia and Emmenegger \cite{coraglia2023categorical} 
propose to view vertical morphisms in a generalized category with families as witnesses for a proof-relevant coercive subtyping. 
They show that every vertical morphism $c : A \to B$ induces a type casting operation transforming terms of type $A$ into terms of type $B$. 
This is essentially the proof-relevant counterpart of the subsumption rule that is usually available in type theories with subtyping. 
They also study how vertical morphisms, and so the type casting operation, interact with the most common type formers. 
Note that, by relying on the equivalence between generalized categories with families and (Jacobs) comprehension categories \cite{emmenegger}, 
these results apply also to the latter ones. 

Adjedj et al.  \cite{adjedj2026adaptt} define AdapTT, a dependent type theory endowed with extra structure in order to provide a general framework to understand type casting operations. 
This work is tightly related to the approach by Coraglia and Emmenegger: 
the semantic model of AdapTT is given using natural models with discrete opfibration, which the authors prove to be equivalent to split generalized categories with families, and so to split comprehension categories.
This means that in the model 
substitution behaves functorially, while in arbitrary comprehension categories it is functorial only up to isomorphism.

Najmaei et al. \cite{najmaei2026semantics} propose a new type theory, dubbed Comprehension Categories Type Theory (CCTT for short), specifically designed to reason synthetically about the structure of comprehension categories.
It is obtained by reflecting semantic features of theory of comprehension categories back into the syntax, notably, vertical morphisms have a dedicated judgement in the syntax. 
In particular models differ from the ones of AdapTT by dropping the splitness requirement: in this sense they give a syntax able to reason about comprehension categories in full generality.

\subsubsection*{Future work} 

We envision several directions for further development.
First, a natural question is whether the bi-adjunctions we have introduced are 2-monadic. More precisely, we already know that all of them induce a pseudo-monad on the domain of the left bi-adjoint, which has an associated 2-category of pseudo-algebras \cite{blackwell1989two}. Proving 2-monadicity then amounts to showing that the codomain of the left bi-adjoint is bi-equivalent to this category of pseudo-algebras. This would ensure that the considered comprehension structures are essentially algebraic concepts, providing us with useful categorical constructions on them (e.g., regarding limits and colimits).

Another natural direction is to study how the construction of the free Lawvere-Ehrhard comprehension category interacts with the type constructors (e.g., $\Pi$-types, $\Sigma$-types, and $\mathsf{Id}$-types) available in the underlying comprehension category. This could help clarify how these type formers interact with the Lawvere-Ehrhard condition and whether their definitions need to be adjusted in this context.

Finally, it would be interesting to develop a syntax for Lawvere-Ehrhard comprehension categories, in the same spirit as the syntax for comprehension categories proposed in \cite{najmaei2026semantics}. This would provide an internal language for Lawvere-Ehrhard comprehension categories, enabling synthetic reasoning about them.

\bibliographystyle{alphaurl} 
\bibliography{bibliography}

\appendix
\section{Proofs of technical results of \zcref{sect:free-jcomp}}\label{app:A}
In this section we prove some technical results concerning \zcref{sect:free-jcomp}.
Here we write $\twont{\eta}$ and $\twont{\epsilon}$ instead of $\Jun{}$ and $\Jcoun{}$.

\subsection{From $\Fib$ to $\JComp$}
\begin{rem}\label{rmk:lemma}
	Given an arrow $f=(g,h):Y=((X,\vec{A}), A_{n+1})\rightarrow Z=((X', \vec{B}), B_{m+1})$, one has $\chi\ft{\twont{\epsilon}_p}(f): \chi A_{n+1}^*\rightarrow \chi B_{m+1}^*$, so it is the square
	\[\begin{tikzcd}
		{\comp{} A_{n+1}^*} & {\comp{} B_{m+1}^*} \\
		{\fb{\twont{\epsilon}_p}(X,\vec{A})} & {\fb{\twont{\epsilon}_p}(X',\vec{B})}
		\arrow["{\comp{} h^*}", from=1-1, to=1-2]
		\arrow["{\chi A_{n+1}^*}"', from=1-1, to=2-1]
		\arrow["{\chi B_{m+1}^*}", from=1-2, to=2-2]
		\arrow["\fb{\twont{\epsilon}_p}g"', from=2-1, to=2-2]
	\end{tikzcd}\]
	
	Analogously, $\Arr{(\fb{\twont{\epsilon}_p})}\JFC{p}(f)$ is the square 
	\[\begin{tikzcd}
		{\fb{\twont{\epsilon}_p}(X,(\vec{A},A_{n+1}))} & {\fb{\twont{\epsilon}_p}((X',(\vec{B},B_{m+1}))} \\
		{\fb{\twont{\epsilon}_p}(X,\vec{A})} & {\fb{\twont{\epsilon}_p}(X',\vec{B})}
		\arrow["{\fb{\twont{\epsilon}_p}\carr{f}}", from=1-1, to=1-2]
		\arrow["{\fb{\twont{\epsilon}_p}\JFC{p} Y}"', from=1-1, to=2-1]
		\arrow["{\fb{\twont{\epsilon}_p}\JFC{p} Z}", from=1-2, to=2-2]
		\arrow["\fb{\twont{\epsilon}_p}g"', from=2-1, to=2-2]
	\end{tikzcd}\]
\end{rem}

So we only need to prove that $\fb{\twont{\epsilon}_p}\JFC{p} Y=\chi A_{n+1}^*$ and that $\comp{} h^*=\fb{\twont{\epsilon}_p}\carr{f}$. These will be shown in the following lemmas.

\begin{lem}\label{lem:jcoun1}
	Let $Y=((X,\vec{A}), A_{n+1})$ in $\JDomE{p}$. Then $\fb{\twont{\epsilon}_p}\JFC{p} Y=\C{n}{n+1}$.
\end{lem}

\begin{proof}
	By induction on $n$.
	\begin{itemize}[align=left]
		\item[\emph{n=0:}] By definition, $\fb{\twont{\epsilon}_p}\JFC{p} Y= \fp{\JFC{p} Y}\arcmp \C{0}{n+1}=\C{0}{1}$;
		\item[\emph{n+1:}] Let $Z=((X,\vec{A}\rstn _n), A_{n+1})$. By inductive hypothesis $\fb{\twont{\epsilon}_p}\JFC{p} Z=\C{n}{n+1}$. 
		By definition $\fb{\twont{\epsilon}_p}\JFC{p} Y$ is the only arrow defined by the universal property of the pullback
		\[\begin{tikzcd}
			&& {} & {\fb{\twont{\epsilon}_p}(X,\vec{A})} \\
			{\fb{\twont{\epsilon}_p}(X,(\vec{A}, A_{n+2}))} && {\comp{} A_{n+1}} & {\fb{\twont{\epsilon}_p}(X,\vec{A}\rstn_n)} \\
			& {\fb{\twont{\epsilon}_p}(X,\vec{A})} & {\comp{} A_{n+1}} & X \\
			& {\fb{\twont{\epsilon}_p}(X,\vec{A}\rstn_n)} & {X}
			\arrow[from=1-4, to=2-3]
			\arrow["{\C{n}{n+1}}", from=1-4, to=2-4]
			\arrow["{\C{n+1}{n+2}}", from=2-1, to=1-4]
			\arrow["\fb{\twont{\epsilon}_p}\JFC{p} Y", dashed, from=2-1, to=3-2]
			\arrow["\fb{\twont{\epsilon}_p}h"', bend right, from=2-1, to=4-2]
			\arrow["{\id{\comp{} A_{n+1}}}"', from=2-3, to=3-3]
			\arrow[from=2-3, to=3-4]
			\arrow["{\C{0}{n}}", from=2-4, to=3-4]
			\arrow[from=3-2, to=3-3]
			\arrow["\chi A_{n+1}^*", from=3-2, to=4-2]
			\arrow["{\chi A_{n+1}}"', from=3-3, to=4-3]
			\arrow["{\id{X}}"', from=3-4, to=4-3]
			\arrow["{\C{0}{n}}", from=4-2, to=4-3]
		\end{tikzcd}\]
		where $h=\JFC{p} Z\arcmp\JFC{p} Y$.
		So $\fb{\twont{\epsilon}_p}h=\fb{\twont{\epsilon}_p}\JFC{p} Z\arcmp \fb{\twont{\epsilon}_p}\JFC{p} Y=\C{n}{n+1}\arcmp\JFC{p} Y$. Since $\chi A_{n+1}^*=\C{n}{n+1}$ one has that $\C{n+1}{n+2}$ makes the left triangle commute. It clearly makes also the right triangle to commute since the arrows are the same, so by the universal property of the pullback we have that $\fb{\twont{\epsilon}_p}\JFC{p} Y=\C{n+1}{n+2}$.
	\end{itemize}
\end{proof}

\begin{lem}\label{lem:jcoun2}
	Let $f = (g,h):((X,\vec{A}), A_{n+1})\rightarrow ((X',\vec{B}), B_{m+1})$ be an arrow in $\JDomE{p}$. Then $\comp{} h^* = \fb{\twont{\epsilon}_p}\carr{f}$.
\end{lem}

\begin{proof}		
	\[\begin{tikzcd}
		&& {} & {\fb{\twont{\epsilon}_p}(X,(\vec{A},A_{n+1}))} \\
		{\fb{\twont{\epsilon}_p}(X,(\vec{A},A_{n+1}))} && {\comp{} A_{n+1}} & {\fb{\twont{\epsilon}_p}(X,\vec{A})} \\
		& {\fb{\twont{\epsilon}_p}(X',(\vec{B}, B_{m+1}))} & {\comp{} B_{m+1}} & X \\
		& {\fb{\twont{\epsilon}_p}(X', \vec{B})} & {X'}
		\arrow[from=1-4, to=2-3]
		\arrow["{\C{n}{n+1}}", from=1-4, to=2-4]
		\arrow["{\C{n+1}{n+1}}", from=2-1, to=1-4]
		\arrow["\fb{\twont{\epsilon}_p}\carr{f}"', dashed, from=2-1, to=3-2]
		\arrow["\fb{\twont{\epsilon}_p}(\JFC{p} Z\arcmp \carr{f})"', bend right, from=2-1, to=4-2]
		\arrow["{\comp{}h}"', from=2-3, to=3-3]
		\arrow[from=2-3, to=3-4]
		\arrow["{\C{0}{n}}", from=2-4, to=3-4]
		\arrow[from=3-2, to=3-3]
		\arrow["\chi B_{m+1}^*", from=3-2, to=4-2]
		\arrow["{\chi B_{m+1}}"', from=3-3, to=4-3]
		\arrow["{\fp{g}}"', from=3-4, to=4-3]
		\arrow["{\C{0}{m}}", from=4-2, to=4-3]
	\end{tikzcd}\]			
	Since $\comp{} h^*$	 makes the right triangle to commute, we only need to show that $\fb{\twont{\epsilon}_p}(\JFC{p}Z\arcmp \carr{f})=\chi B_{m+1}^* \arcmp (\comp{} h^*)$. But one has $\fb{\twont{\epsilon}_p}(\JFC{p}Z\arcmp \carr{f})=\fb{\twont{\epsilon}_p}(g\arcmp \JFC{p}Y)=\fb{\twont{\epsilon}_p}(g)\arcmp \chi A_{n+1}^*=\chi B_{m+1}^* \arcmp (\comp{} h^*)$, where the first and the third equalities hold by commutativity of the two diagrams in \zcref{rmk:lemma}. Then by the universal property of the pullback one has that $\fb{\twont{\epsilon}_p}\carr{f}=\comp{} h^*$. 
\end{proof}

\begin{prop}\label{prop:junit}
	$\twont{\eta}$ is a 2-natural transformation.
\end{prop}

\begin{proof}
	Let $F:p\rightarrow q$ be a fibration morphism, and consider the diagram
	\[\begin{tikzcd}
		& {\ct{E'}} && {\JDom{\ct{E}'}{q}} \\
		{\ct{E}} && {\JDomE{p}} \\
		& {\ct{B'}} && {\FFFP{q}} \\
		{\ct{B}} && {\FFFP{p}}
		\arrow[""{name=0, anchor=center, inner sep=0}, "{{\ft{\twont{\eta}_{q}}}}", from=1-2, to=1-4]
		\arrow["{{q}}"{pos=0.3}, from=1-2, to=3-2]
		\arrow["{{\JFib{q}}}", from=1-4, to=3-4]
		\arrow["{\ft{F}}", from=2-1, to=1-2]
		\arrow["{{\ft{\twont{\eta}_p}}}"{pos=0.8}, from=2-1, to=2-3, crossing over]
		\arrow["p", from=2-1, to=4-1]
		\arrow[""{name=1, anchor=center, inner sep=0}, "{{\ft{\JFun{F}}}}"', from=2-3, to=1-4]
		\arrow[""{name=2, anchor=center, inner sep=0}, "{{\fb{\twont{\eta}_{q}}}}"{pos=0.2}, from=3-2, to=3-4]
		\arrow["{{\JFib{p}}}"'{pos=0.3}, from=2-3, to=4-3, crossing over]
		\arrow["{\fb{F}}"{pos=0.7}, from=4-1, to=3-2]
		\arrow["{\fb{\twont{\eta}_p}}"', from=4-1, to=4-3]
		\arrow[""{name=3, anchor=center, inner sep=0}, "{{\fb{\JFun{F}}}}"', from=4-3, to=3-4]
		\arrow["\ft{\twont{\eta}_F}"', bend left, shorten <=7pt, shorten >=7pt, Rightarrow, from=1, to=0]
		\arrow["\fb{\twont{\eta}_F}"', bend left, shorten <=5pt, shorten >=5pt, Rightarrow, from=3, to=2]
	\end{tikzcd}\]
	
	It is easy to see that $\twont{\eta}_q\arcmp F=\JFun{F}\arcmp\twont{\eta}_p$.
\end{proof}

\begin{prop}\label{prop:jcounit}
	$\twont{\epsilon}$ is a pseudo-natural transformation.
\end{prop}

\begin{proof}
	Let $F:p\rightarrow q$ together with $\alpha:\Arr{(\fb{F})}\arcmp \jcomp{p} \nt \jcomp{q}\arcmp \ft{F}$ be a morphism in $\JComp$, and consider the following diagram:
	\[\begin{tikzcd}
		& {\JDom{\ct{E}'}{q}} &&& {\ct{E}'} \\
		{\JDomE{p}} && {\ct{E}} \\
		& {\FFFP{q}} &&& {\ct{B'}} \\
		{\FFFP{p}} && {\ct{B}}
		\arrow[""{name=0, anchor=center, inner sep=0}, "{{{\ft{\Jcoun{q}}}}}", from=1-2, to=1-5]
		\arrow["{{{\JFib{q}}}}"'{pos=0.3}, from=1-2, to=3-2]
		\arrow["{{{q}}}", from=1-5, to=3-5]
		\arrow["{{{\ft{\JFun{F}}}}}", from=2-1, to=1-2]
		\arrow["{{{\ft{\twont{\epsilon}_p}}}}"{pos=0.8}, from=2-1, to=2-3, crossing over]
		\arrow["{{{\JFib{p}}}}"', from=2-1, to=4-1]
		\arrow[""{name=1, anchor=center, inner sep=0}, "{{\ft{F}}}"', from=2-3, to=1-5]
		\arrow[""{name=2, anchor=center, inner sep=0}, "{{{\fb{\Jcoun{q}}}}}"{pos=0.8}, from=3-2, to=3-5]
		\arrow["p"'{pos=0.6}, from=2-3, to=4-3, crossing over]
		\arrow["{{{\fb{\JFun{F}}}}}"{pos=0.7}, from=4-1, to=3-2]
		\arrow["{{\fb{\Jcoun{p}}}}"', from=4-1, to=4-3]
		\arrow[""{name=3, anchor=center, inner sep=0}, "{{\fb{F}}}"', from=4-3, to=3-5]
		\arrow["{\ft{\Jcoun{F}}}"', bend left, shorten <=4pt, shorten >=4pt, Rightarrow, from=1, to=0]
		\arrow["{\fb{\Jcoun{F}}}"', bend left, shorten <=4pt, shorten >=4pt, Rightarrow, from=3, to=2]
	\end{tikzcd}\]
	We have that $\Jcoun{F}$ is an invertible 2-cell by construction. With routine calculations can be shown that the coherence conditions required for the pseudo-naturality are satisfied.
\end{proof}

\begin{prop}\label{prop:trian-id-1}
	Let $p:\ct{E}\to \ct{B}$ together with $\jcomp{p}:\ct{E}\rightarrow \Arr{\ct{B}}$ be a comprehension category, and consider the diagram
	\[\begin{tikzcd}
		{\ct{E}} & {\JDomE{p}} & {\ct{E}} \\
		{\ct{B}} & {\FFFP{p}} & {\ct{B}}
		\arrow["{{\ft{\twont{\eta}_p}}}"', from=1-1, to=1-2]
		\arrow[""{name=0, anchor=center, inner sep=0}, "{{{\Id{\ct{E}}}}}", bend left, from=1-1, to=1-3]
		\arrow["p"', from=1-1, to=2-1]
		\arrow["{{\ft{\twont{\epsilon}_p}}}"', from=1-2, to=1-3]
		\arrow["{{{\JFib{p}}}}", from=1-2, to=2-2]
		\arrow["p"', from=1-3, to=2-3]
		\arrow["{\fb{\twont{\eta}_p}}", from=2-1, to=2-2]
		\arrow[""{name=1, anchor=center, inner sep=0}, "{{{\Id{\ct{B}}}}}"', bend right, from=2-1, to=2-3]
		\arrow["{\fb{\twont{\epsilon}_p}}", from=2-2, to=2-3]
		\arrow["{\ft{\alpha_p}}"', shorten >=3pt, Rightarrow, from=1-2, to=0]
		\arrow["{\fb{\alpha_p}}", shorten >=3pt, Rightarrow, from=2-2, to=1]
	\end{tikzcd}\]
	where $\fb{\alpha_p}\colon = \idtwo{\Id{\ct{B}}}$ is the identity natural transformation and $\ft{\alpha_p}\colon = \delta^{-1}$, with $\delta:\Id{\ct{E}}\nt \reind{\id{}}$ the natural isomorphism obtained by reindexing along the identity. Then $\alpha:(\JFree\arcmp \forgJ)\nt \twoId{\JComp}$ is an invertible modification.
\end{prop}

\begin{proof}
	First, we need to show that $\fb{\twont{\epsilon}_p}\arcmp\fb{\twont{\eta}_p}=\Id{\ct{B}}$. This is just a straightforward consequence of the definitions of $\fb{\twont{\eta}_p}$ and the base case of $\fb{\twont{\epsilon}_p}$.
	
	Afterwards, it is enough to show that $\delta^{-1}:\ft{\twont{\epsilon}_p}\arcmp\ft{\twont{\eta}_p}\nt\Id{\ct{E}}$, since $\delta^{-1}$ is trivially invertible. But this is again obvious by their definition: $\ft{\twont{\epsilon}_p}\ft{\twont{\eta}_p}A=\ft{\twont{\epsilon}_p}((pA,()),A)=\reind{(\C{0}{0})}A=\reind{\id{(pA)}}A$.
	
	Finally, the naturality of $\alpha$ with respect to the 1-cells is the result of a routine calculation. 
\end{proof}

\begin{prop}\label{prop:trian-id-2}
	Let $p:\ct{E}\rightarrow\ct{B}$ be a fibration, and consider the diagram
	\[\begin{tikzcd}
		{\JDomE{p}} & {\JDom{\FFFP{p}}{\JFib{p}}} & {\JDomE{p}} \\
		{\FFFP{p}} & {\FFFP{\JFib{p}}} & {\FFFP{p}}
		\arrow["{\ft{\JFun{\twont{\eta}_p}}}", from=1-1, to=1-2]
		\arrow[""{name=0, anchor=center, inner sep=0}, "{{\Id{\JDomE{p}}}}", bend left, from=1-1, to=1-3]
		\arrow["{\JFib{p}}"', from=1-1, to=2-1]
		\arrow["{\ft{\twont{\epsilon}_{\JFun{p}}}}", from=1-2, to=1-3]
		\arrow["{{\JFib{\JFib{p}}}}", from=1-2, to=2-2]
		\arrow["{{\JFib{p}}}", from=1-3, to=2-3]
		\arrow["{\fb{\JFun{\twont{\eta}_p}}}"', from=2-1, to=2-2]
		\arrow[""{name=1, anchor=center, inner sep=0}, "{{\Id{\Arr{\FFFP{p}}}}}"', bend right, from=2-1, to=2-3]
		\arrow["{\fb{\twont{\epsilon}_{\JFun{p}}}}"', from=2-2, to=2-3]
		\arrow["{\ft{\beta_p}}"', shorten >=3pt, Rightarrow, from=1-2, to=0]
		\arrow["{\fb{\beta_p}}", shorten >=3pt, Rightarrow, from=2-2, to=1]
	\end{tikzcd}\]
	where $\beta_p\colon= \id{\JFib{p}}$. Then $\beta$ is an invertible modification.
\end{prop}

\begin{proof}
	We only need to show that $\twont{\epsilon}_{\JFun{p}}\arcmp \JFun{\twont{\eta}_p}=\id{\JFib{p}}$. This is a straightforward consequence of the definition of $\twont{\eta}$ and $\twont{\epsilon}$.
	
\end{proof}

\section{Proof of technical results of \zcref{sect:JtoL}}\label{app:B}

Here we prove some technical results about \zcref{sect:JtoL}.

\subsection{From $\JComp$ to $\JTComp$}
\begin{lem}\label{lem:JTun-pseudonat}
	$\JTun{}$ is a 2-natural transformation.
\end{lem}

\begin{proof}
	Consider a morphism $F:p\to q$ in $\JComp$. It is straightforward that $\JTFun{F}\arcmp \JTun{p}=\JTun{q}\arcmp F$.
\end{proof}

\begin{lem}\label{lem:JTcoun-pseudonat}
	$\JTcoun{}$ is a pseudo-natural transformation.
\end{lem}

\begin{proof}
	Consider a morphism $F:p\to q$ in $\JTComp$. The 2-cell $\JTcoun{F}$ is invertible since it is the iso testifying that $F$ preserves fibred terminals. Its naturality follows from a routine calculation.
\end{proof}

\subsection{From $\JTComp$ to $\LComp$}

\begin{lem}\label{lem:LEun-nat}
	$\LEun{}:\twoId{\JTComp}\to \forgLE\arcmp \LEFree$ is a 2-natural transformation.
\end{lem}

\begin{proof}
	Consider a 1-cell $F:p\to q$ in $\JTComp$. Then the following diagram commutes, proving the claim.
	\[\begin{tikzcd}
		{\ct{E}} && {\Tr(p)} && {\LEDomE{p}} \\
		\\
		{\ct{E'}} && {\Tr(q)} && {\LEDom{\ct{E}'}{q}}
		\arrow[hook, from=1-1, to=1-3]
		\arrow["{\ft{F}}"{description}, from=1-1, to=3-1]
		\arrow["{{\quot{p}}}"{description}, from=1-3, to=1-5]
		\arrow["{\Tr(F)}"{description}, from=1-3, to=3-3]
		\arrow["{\LEFun{F}}"{description}, dashed, from=1-5, to=3-5]
		\arrow[hook, from=3-1, to=3-3]
		\arrow["{{\quot{q}}}"{description}, from=3-3, to=3-5]
	\end{tikzcd}\]
\end{proof}

\begin{lem}\label{lem:LEcoun-nat}
	$\LEcoun{}:\LEFree\arcmp\forgLE\to \twoId{\LComp}$ is a 2-natural transformation.
\end{lem}

\begin{proof}
	It is natural since it is determined by a universal property.
\end{proof}

		\end{document}